\RequirePackage{snapshot}
\documentclass[a4paper,UKenglish]{article}

\usepackage{microtype}
\usepackage{todonotes}
\usepackage{amsthm}
\usepackage{xspace}
\usepackage[ruled,linesnumbered, algo2e]{algorithm2e}
\usepackage{float}

\title{KADABRA is an ADaptive Algorithm for Betweenness via Random Approximation\footnote{This work was done while the authors were visiting the Simons Institute for the Theory of Computing.}}

\usepackage{authblk}
\usepackage{amsmath,amsfonts,amsopn,amsbsy}
\usepackage[hidelinks]{hyperref}
\usepackage[margin=1.4in]{geometry}
\usepackage{babel}
\usepackage{cleveref}

\author[1]{Michele Borassi}
\author[2]{Emanuele Natale}
\affil[1]{IMT Insitute for Advanced Studies, 55100 Lucca, Italy\\
  \texttt{michele.borassi@imtlucca.it}}
\affil[2]{Sapienza University of Rome, 00185 Roma, Italy\\
  \texttt{natale@di.uniroma1.it}}

\renewcommand{\P}{\Pr}
\newcommand{\p}{\boldsymbol{\pi}}
\newcommand{\Path}{\boldsymbol{\pi}}
\newcommand{\E}{\mathbb{E}}
\newcommand{\Xvp}[2]{X_{#1}^{#2}}
\newcommand{\dep}[2]{\delta_{#1}(#2)}

\newcommand{\sig}[2]{\sigma_{#1#2}}
\newcommand{\sigv}[3]{\sigma_{#1#2}(#3)}

\newcommand{\ba}[1]{\boldsymbol{b}(#1)}

\renewcommand{\O}{\mathcal{O}}
\newcommand{\w}[1]{\rho_{#1}}

\newcommand{\whp}{w.h.p.}
\newcommand{\aas}{a.a.s.}

\newcommand{\N}[2]{\boldsymbol{N}^{#1}(#2)}

\newcommand{\G}[2]{\boldsymbol{\Gamma}^{#1}(#2)}

\newcommand{\R}[2]{\boldsymbol{R}^{#1}(#2)}
\renewcommand{\r}[2]{\boldsymbol{r}^{#1}(#2)}
\newcommand{\wres}{\rho_{\text{\upshape res}}}

\newcommand{\X}{\boldsymbol{X}}
\newcommand{\Y}{\boldsymbol{Y}}
\newcommand{\Z}{\boldsymbol{Z}}

\newcommand{\stoptime}{\boldsymbol{\tau}}
\newcommand{\fixtime}{\tau}
\newcommand{\deltal}{\delta_{L}}
\newcommand{\deltau}{\delta_{U}}
\newcommand{\lambdal}{\lambda_{L}}
\newcommand{\lambdau}{\lambda_{U}}
\newcommand{\logdl}{\log{\frac 1{\deltal}}}
\newcommand{\logdu}{\log{\frac 1{\deltau}}}
\newcommand{\btilde}{\tilde{\boldsymbol{b}}}
\newcommand{\btildev}[1]{\btilde(#1)}

\newcommand{\btildekv}[2]{\btilde_{#1}(#2)}
\newcommand{\betv}[1]{\bet(#1)}
\newcommand{\bet}{\bc}

\newcommand{\deltalv}[1]{\delta_{L}^{(#1)}}
\newcommand{\deltauv}[1]{\delta_{U}^{(#1)}}
\newcommand{\lambdalv}[1]{\lambda_{L}^{(#1)}}
\newcommand{\lambdauv}[1]{\lambda_{U}^{(#1)}}
\newcommand{\logdlv}[1]{\log{\frac 1{\deltalv{#1}}}}
\newcommand{\logduv}[1]{\log{\frac 1{\deltauv{#1}}}}

\newcommand{\ETOP}{\widetilde{TOP}}

\newcommand{\rk}{{RK}}
\newcommand{\abra}{{ABRA}}
\newcommand{\abraaut}{{ABRA-Aut}}
\newcommand{\abrag}{{ABRA-1.2}}
\newcommand{\cad}{{KADABRA}\xspace}
\newcommand{\bbbfs}{bb-BFS\xspace}

\newcommand{\good}{\boldsymbol{A}}

\DeclareMathOperator{\vd}{VD}
\DeclareMathOperator{\bc}{bc}

\DeclareMathOperator{\res}{res}
\DeclareMathOperator{\var}{Var}
\DeclareMathOperator{\avg}{avg}

\newtheorem{thm}{Theorem}
\newtheorem{lem}[thm]{Lemma}

\newtheorem{remark}[thm]{Remark}
\newtheorem{definition}[thm]{Definition}

\crefname{axiom}{Axiom}{Axioms}
\crefname{thm}{Theorem}{Theorems}
\crefname{claim}{Claim}{Claims}
\crefname{lem}{Lemma}{Lemmas}
\crefname{fact}{Fact}{Facts}
\crefname{proposition}{Proposition}{Propositions}
\crefname{cor}{Corollary}{Corollaries}

\begin{document}

\maketitle

\begin{abstract}
We present \cad, a new algorithm to approximate betweenness centrality in directed and undirected graphs, which significantly outperforms all previous approaches on real-world complex networks. 
The efficiency of the new algorithm relies on two new theoretical contributions, of independent interest.

The first contribution focuses on sampling shortest paths, a subroutine used by most algorithms that approximate betweenness centrality. We show that, on realistic random graph models, we can perform this task in time $|E|^{\frac{1}{2}+o(1)}$ with high probability, obtaining a significant speedup with respect to the $\Theta(|E|)$ worst-case performance. We experimentally show that this new technique achieves similar speedups on real-world complex networks, as well.

The second contribution is a new rigorous application of the adaptive sampling technique. This approach decreases the total number of shortest paths that need to be sampled to compute all betweenness centralities with a given absolute error, and it also handles more general problems, such as computing the $k$ most central nodes. Furthermore, our analysis is general, and it might be extended to other settings, as well.
\end{abstract}

\section{Introduction}

In this work we focus on estimating the \emph{betweenness centrality}, which is one of the most famous measures of \emph{centrality} for nodes and edges of real-world complex networks \cite{Easley2010NetworksCA,newman2010networks}. 
The rigorous definition of betweenness centrality has its roots in sociology, dating back to the Seventies, when Freeman formalized the informal concept discussed in the previous decades in different scientific communities \cite{bavelas1948mathematical,shimbel1953structural,shaw1954group,cohn1958networks,Borgatti2006AGP}, although the definition already appeared in \cite{anthonisse1971rush}. 
Since then, this notion has been very successful in network science \cite{wasserman1994social,newman2001scientific,Geisberger2008ContractionHF,newman2010networks}. 

A probabilistic way to define the betweenness centrality\footnote{As explained in see Section \ref{sec:algoshort}, to simplify notation we consider the \emph{normalized} betweenness centrality.} $\betv{v}$ of a node $v$ in a graph $G=(V,E)$ is the following. We choose two nodes $s$ and $t$, and we go from $s$ to $t$ through a shortest path $\pi$; if the choices of $s$, $t$ and $\pi$ are made uniformly at random, the betweenness centrality of a node $v$ is the probability that we pass through $v$.

In a seminal paper \cite{Brandes2001}, Brandes showed that it is possible to exactly compute the betweenness centrality of all the nodes in a graph in time $\O(mn)$, where $n$ is the number of nodes and $m$ is the number of edges. A corresponding lower bound was proved in \cite{Borassi2015}: if we are able to compute the betweenness centrality of a single node in time $\O(mn^{1-\epsilon})$ for some $\epsilon>0$, then the Strong Exponential Time Hypothesis \cite{Impagliazzo2001} is false.

This result further motivates the rich line of research on computing approximations of betweenness centrality, with the goal of trading precision with efficiency.
The main idea is to 
define a probability distribution over the set of all paths, by choosing two uniformly random nodes $s,t$, and then a uniformly distributed $st$-path $\Path$, 
so that $\Pr(v \in \Path)=\betv v$. As a consequence, we can approximate $\betv v$ by sampling paths $\Path_1,\dots,\Path_\tau$ according to this distribution, and estimating $\btildev v:=\frac{1}{\fixtime}\sum_{i=1}^\fixtime \X_i(v)$, where $\X_i(v)=1$ if $v \in \Path_i$ (and $v \neq s,t$), $0$ otherwise.

The tricky part of this approach is to provide probabilistic guarantees on the quality of this approximation: the goal is to obtain a $1-\delta$ confidence interval $\boldsymbol{I}(v)=[\btildev v-\lambdal,\btildev v+\lambdau]$ for $\betv v$, which means that $\Pr(\forall v \in V,\betv v \in \boldsymbol{I}(v))\geq 1-\delta$. 
Thus, the research for approximating betweenness centrality has been focusing on obtaining, as fast as possible, the smallest possible $\boldsymbol{I}$.

\subsubsection*{Our Contribution}
\label{sec:contribshort}

In this work, we propose a new and faster algorithm to approximate betweenness centrality in directed and undirected graphs, named \cad. In the standard task of approximating betweenness centralities with absolute error at most $\lambda$, we show that, on average, the new algorithm is more than $100$ times faster than the previous ones, on graphs with approximately $10\,000$ nodes. Moreover, differently from previous approaches, our algorithm can perform more general tasks, since it does not need all confidence intervals to be equal. As an example, we consider the computation of the $k$ most central nodes: all previous approaches compute all centralities with an error $\lambda$, and use this approximation to obtain the ranking. Conversely, our approach allows us to use small confidence interval only when they are needed, and allows bigger confidence intervals for nodes whose centrality values are ``well separated''. This way, we can compute for the first time an approximation of the $k$ most central nodes in networks with millions of nodes and hundreds of millions of edges, like the Wikipedia citation network and the IMDB actor collaboration network. 

Our results rely on two main theoretical contributions, which are interesting in their own right, since their generality naturally extends to other applications. 

\subparagraph{Balanced bidirectional breadth-first search.}
By leveraging on recent advanced results, we prove that, on many realistic random models of real-world complex networks, it is possible to sample a random path between two nodes $s$ and $t$ in time $m^{\frac{1}{2}+o(1)}$ if the degree distribution has finite second moment, or $m^{\frac{4-\beta}{2}+o(1)}$ if the degree distribution is power law with exponent $2<\beta<3$. The models considered are the Configuration Model \cite{Bollobas1980}, and all Rank-1 Inhomogeneous Random Graph models \cite[Chapter 3]{Hofstad2014}, such as the Chung-Lu model \cite{Chung2006}, the Norros-Reittu model \cite{Norros2006}, and the Generalized Random Graph \cite[Chapter 3]{Hofstad2014}.
Our proof techniques have the merit of adopting a unified approach that simultaneously works in all models considered. These models well represent metric properties of real-world networks \cite{Borassi2016}: indeed, our results are confirmed by practical experiments.

The algorithm used is simply a balanced bidirectional BFS (\bbbfs): we perform a BFS from each of the two endpoints $s$ and $t$, in such a way that the two BFSs are likely to explore about the same number of edges, and we stop as soon as the two BFSs ``touch each other''. 
Rather surprisingly, this technique was never implemented to approximate betweenness centrality, and it is rarely used in the experimental algorithm community. 
Our theoretical analysis provides a clear explanation of the reason why this technique improves over the standard BFS: this means that many state-of-the-art algorithm for real-world complex networks can be improved by the \bbbfs.

\subparagraph{Adaptive sampling made rigorous.}
To speed up the estimation of the betweenness centrality, previous work make use of the technique of adaptive sampling, which consists in testing during the execution of the algorithm whether some condition on the sample obtained so far has been met, and terminating the execution of the algorithm as soon as this happens. 
However, this technique introduces a subtle stochastic dependence between the time in which the algorithm terminates and the correctness of the given output, which previous papers claiming a formal analysis of the technique did not realize (see Section \ref{sec:adapshort} for details). 
With an argument based on martingale theory, we provide a general analysis of such useful technique. Through this result, we do not only improve previous estimators, but we also make it possible to define more general stopping conditions, that can be decided ``on the fly'': this way, with little modifications, we can adapt our algorithm to perform more general tasks than previous ones.

To better illustrate the power of our techniques, we focus on the unweighted, static graphs, and to the centrality of nodes. 
However, our algorithm can be easily adapted to compute the centrality of edges, to handle weighted graphs and, since its core part consists merely in sampling paths, we conjecture that it may be coupled with the existing techniques in \cite{Bergamini2015FullyDynamicAO} to handle dynamic graphs. 

\subsubsection*{Related Work}
\label{sec:relatedshort}

\subparagraph{Computing Betweenness Centrality.}
With the recent event of big data, the major shortcoming of betweenness centrality has been the lack of efficient methods to compute it \cite{Brandes2001}. 
In the worst case, the best exact algorithm to compute the centrality of all the nodes is due to Brandes \cite{Brandes2001}, and its time complexity is $\O(mn)$: the basic idea of the algorithm is to define the dependency $\dep sv=\sum_{t \in V} \frac{\sigv stv}{\sig st}$, which can be computed in time $\O(m)$, for each $v \in V$ (we denote by $\sigv stv$ the number of shortest paths from $s$ to $t$ passing through $v$, and by $\sig st$ the number of $st$-shortest paths). 
In \cite{Borassi2015}, it is also shown that Brandes algorithm is almost optimal on sparse graphs: an algorithm that computes the betweenness centrality of a single vertex in time $\O(mn^{1-\epsilon})$ falsifies widely believed complexity assumptions, such as the Strong Exponential Time Hypothesis \cite{Impagliazzo2001}, the Orthogonal Vector conjecture \cite{Abboud2016}, or the Hitting Set conjecture \cite{Williams2014}. Corresponding results in the dense, weighted case are available in \cite{grand2015}: computing the betweenness centrality exactly is as hard as computing the All Pairs Shortest Path, and computing an approximation with a given relative error is as hard as computing the diameter. For both these problems, there is no algorithm with running-time $\O(n^{3-\epsilon})$, for any $\epsilon>0$. This shows that, for dense graphs, having an additive approximation rather than a multiplicative one is essential for a provably fast algorithm to exist.
These negative results further motivate the already rich line of research on approaches that overcome this barrier. A first possibility is to use heuristics, that do not provide analytical guarantees on their performance \cite{atalyrek2013ShatteringAC,Erds2015ADA,Vella2016AlgorithmsAH}. Another line of research has defined variants of betweenness centrality, that might be easier to compute \cite{Brandes2008OnVO,pfeffer2012k,Dolev2010RoutingBC}. Finally, a third line of research has investigated approximation algorithms, which trade accuracy for speed \cite{Jacob2004AlgorithmsFC,Brandes2007,Geisberger2008ContractionHF,lim2011online}. Our work follows the latter approach.
The first approximation algorithm proposed in the literature \cite{Jacob2004AlgorithmsFC} adapts Eppstein and Wang's approach for computing closeness centrality \cite{Eppstein2001FastAO}, using Hoeffding's inequality and the union bound technique. This way, it is possible to obtain an estimate of the betweenness centrality of every node that is correct up to an additive error $\lambda$ with probability $\delta$, by sampling $\O(\frac{D^2}{\lambda^2}\log\frac{n}{\delta})$ nodes, where $D$ is the diameter of the graph. In
\cite{Geisberger2008ContractionHF}, it is shown that this can lead to an overestimation. 
Riondato and Kornaropoulos improve this sampling-based approach by sampling single shortest paths instead of the whole dependency of a node \cite{Riondato2015}, introducing the use of the VC-dimension.
As a result, the number of samples is decreased to $\frac{c}{\lambda^2}(\lfloor\log_2(\vd-2)\rfloor+1+\log(\frac{1}{\delta}))$, where $\vd$ is the vertex diameter, that is, the minimum number of nodes in a shortest path in $G$ (it can be different from $D+1$ if the graph is weighted). 
This use of the VC-dimension is further developed and generalized in \cite{Riondato2016}. Finally, many of these results were adapted to handle dynamic networks \cite{Bergamini2015FullyDynamicAO,Riondato2016}. 

\subparagraph{Approximating the top-\texorpdfstring{$k$}{k} betweenness centrality set.}
Let us order the nodes $v_1,...,v_n$ such that $\bc(v_1)\geq ...\geq \bc(v_n)$ and define 
$TOP(k) = \{ (v_i,\bc(v_i)): i \leq k\}$. In \cite{Riondato2015} and \cite{Riondato2016}, the authors provide an algorithm that, for any given $\delta,\epsilon$, with probability $1-\delta$ outputs a set $\ETOP(k) = \{(v_i,\btildev{v_i})\}$ such that: 
i) If $v \in TOP(k)$ then $v \in \ETOP(k)$ 
and $| \bc(v) - \btildev{v} | \leq \epsilon \bc(v)$;
ii) If $v \in \ETOP(k)$ but $v \not\in TOP(k)$ then $\btildev{v}\leq (\mathbf{b}_k-\epsilon) (1+\epsilon)$ where $\mathbf{b}_k$ is the $k$-th largest betweenness given by a preliminary phase of the algorithm.

\subparagraph{Adaptive sampling.}
In \cite{Bader2007,Riondato2016}, the number of samples required is substantially reduced using the adaptive sampling technique introduced by Lipton and Naughton in \cite{lipton_estimating_1989,lipton_query_1995}. Let us clarify that, by adaptive sampling, we mean that the termination of the sampling process depends on the sample observed so far (in other cases, the same expression refers to the fact that the distribution of the new samples is a function of the previous ones \cite{aggarwal_adaptive_2009}, while the sample size is fixed in advance). Except for \cite{pietracaprina_mining_2010}, previous approaches tacitly assume that there is little dependency between the stopping time and the correctness of the output: indeed, they prove that, for each \emph{fixed} $\fixtime$, the probability that the estimate is wrong at time $\fixtime$ is below $\delta$. However, the stopping time $\stoptime$ is a random variable, and in principle there might be dependency between the event $\stoptime=\tau$ and the event that the estimate is correct at time $\tau$. 
As for \cite{pietracaprina_mining_2010}, they consider a specific stopping condition and their proof technique does not seem to extend to other settings.
For a more thorough discussion of this issue, we defer the reader to \Cref{sec:adapshort}.

\subparagraph{Bidirectional BFS.}
The possibility of speeding up a breadth-first search for the shortest-path problem by performing, at the same time, a BFS from the final end-point, has been considered since the Seventies \cite{pohl1969bi}.
Unfortunately, because of the lack of theoretical results dealing with its efficiency, the bidirectional BFS has apparently not been considered a fundamental heuristic improvement \cite{Kaindl1997BidirectionalHS}. 
However, in \cite{Riondato2015} (and in some public talks by M. Riondato), the bidirectional BFS was proposed as a possible way to improve the performance of betweenness centrality approximation algorithms.

\subsubsection*{Structure of the Paper}

In \Cref{sec:algoshort}, we describe our algorithm, and in \Cref{sec:adapshort} we discuss the main difficulty of the adaptive sampling, and the reasons why our techniques are not affected. In \Cref{sec:bfsshort}, we define the balanced bidirectional BFS, and we sketch the proof of its efficiency on random graphs. In \Cref{sec:topkshort}, we show that our algorithm can be adapted to compute the $k$ most central nodes. In \Cref{sec:experimentsshort} we experimentally show the effectiveness of our new algorithm. Finally, all our proofs are in the appendix.

\section{Algorithm Overview}
\label{sec:algoshort}

To simplify notation, we always consider the \emph{normalized} betweenness centrality of a node $v$, which is defined by:
\[
    \bc(v)=\frac{1}{n(n-1)}\sum_{s \neq v \neq t} \frac{\sigv stv}{\sig st}
\]
where $\sig st$ is the number of shortest paths between $s$ and $t$, and $\sigv stv$ is the number of shortest paths between $s$ and $t$ that pass through $v$. Furthermore, to simplify the exposition, we use bold symbols to denote random variables, and light symbols to denote deterministic quantities.
On the same line of previous works, our algorithm samples random paths $\p_1,\dots,\p_\fixtime$, where $\p_i$ is chosen by selecting uniformly at random two nodes $s,t$, and then selecting uniformly at random one of the shortest paths from $s$ to $t$. Then, it estimates $\betv v$ with $\btildev v:=\frac{1}{\fixtime}\sum_{i=1}^\fixtime \X_i(v)$, where $\X_i(v)=1$ if $v \in \p_i$, $0$ otherwise. By definition of $\p_i$, $\E\left[\btildev v\right]=\betv v$.

The tricky part is to bound the distance between $\btildev v$ and its expected value. With a straightforward application of Hoeffding's inequality (Lemma~\ref{lem:hoeff} in the appendix), it is possible to prove that $\Pr\left(\left|\btildev v-\betv v\right|\geq \lambda\right) \leq 2e^{-2\fixtime\lambda^2}$. A direct application of this inequality considers a union bound on all possible nodes $v$, obtaining $\Pr(\exists v \in V, |\btildev v-\betv v|\geq \lambda) \leq 2ne^{-2\fixtime\lambda^2}$. This means that the algorithm can safely stop as soon as $2ne^{-2\fixtime\lambda^2} \leq \delta$, that is, after $\fixtime= \frac{1}{2\lambda^2} \log(\frac{2n}{\delta})$ steps. 

In order to improve this idea, we can start from Lemma \ref{lem:cb} in the appendix, instead of Hoeffding inequality, obtaining that $\Pr\left(\left|\btildev v-\betv v\right|\geq \lambda\right) \leq 2\exp(-\frac{\fixtime\lambda^2}{2(\betv v+\lambda/3)})$.

If we assume the error $\lambda$ to be small, this inequality is stronger than the previous one for all values of $\betv v<\frac{1}{4}$ (a condition which holds for almost all nodes, in almost all graphs considered). However, in order to apply this inequality, we have to deal with the fact that we do not know $\betv v$ in advance, and hence we do not know when to stop. Intuitively, to solve this problem, we make a ``change of variable'', and we rewrite the previous inequality as 
\begin{equation}\label{eq:chernoffmodshort}
\begin{aligned} 
\Pr\left(\betv v \leq \btildev v-f  \right)\leq \deltalv v \quad \text{and} \quad
\Pr\left(\betv v \geq \btildev v+g\right) \leq \deltauv v,
\end{aligned}
\end{equation}
for some functions $f=f(\btildev v,\deltalv v, \fixtime),g=g(\btildev v,\deltauv v, \fixtime)$. Our algorithm fixes at the beginning the values $\deltalv v, \deltauv v$ for each node $v$, and, at each step, it tests if $f(\btildev v,\deltalv v,\tau)$ and $g(\btildev v,\deltauv v,\tau)$ are small enough. If this condition is satisfied, the algorithm stops. Note that this approach lets us define very general stopping conditions, that might depend on the centralities computed until now, on the single nodes, and so on.

\begin{remark}
Instead of fixing the values $\deltalv v,\deltauv v$ at the beginning, one might want to decide them during the algorithm, depending on the outcome. However, this is not formally correct, because of dependency issues (for example, \eqref{eq:chernoffmodshort} does not even make sense, if $\deltalv v,\deltauv v$ are random). Finding a way to overcome this issue is left as a challenging open problem (more details are provided in \Cref{sec:adapshort}).
\end{remark}

In order to implement this idea, we still need to solve an issue: \eqref{eq:chernoffmodshort} holds for each \emph{fixed} time $\tau$, but the stopping time of our algorithm is a random variable $\stoptime$, and there might be dependency between the value of $\stoptime$ and the probability in \eqref{eq:chernoffmodshort}. To this purpose, we use a stronger inequality (\Cref{thm:mcdiarmid} in the appendix), that holds even if $\stoptime$ is a random variable. However, to use this inequality, we need to assume that $\stoptime < \omega$ for some deterministic $\omega$: in our algorithm, we choose $\omega=\frac{c}{\lambda^2}\left(\lfloor\log_2(\vd-2)\rfloor+1+\log\left(\frac{2}{\delta}\right)\right)$, because, by the results in \cite{Riondato2015}, after $\omega$ samples, the maximum error is at most $\lambda$, with probability $1-\frac{\delta}{2}$. Furthermore, also $f$ and $g$ should be modified, since they now depend on the value of $\omega$. The pseudocode of the algorithm obtained is available in Algorithm~\ref{alg:mainshort} (as was done in previous approaches, we can easily parallelize the while loop in Line~\ref{line:stopcondshort}).

\begin{algorithm2e}
 \begin{footnotesize}
\LinesNumbered
\SetKwFunction{testStoppingCondition}{haveToStop}
\SetKwFunction{computeDeltaV}{computeDelta}
\SetKwFunction{enqueue}{enqueue}
\SetKwFunction{extractMin}{extractMin}
\SetKwFunction{samplePath}{samplePath}
 \SetKwInOut{Input}{Input}
 \SetKwInOut{Output}{Output}
\Input{a graph $G=(V,E)$}
\Output{for each $v\in V$, an approximation $\btildev v$ of $\betv v$ such that $\Pr\left(\forall v,|\btildev v-\betv v|\leq\lambda\right)\geq 1-\delta$}
$\omega \gets \frac{c}{\lambda^2}\left(\lfloor\log_2(\vd-2)\rfloor+1+\log\left(\frac{2}{\delta}\right)\right)$\;

$(\deltalv v,\deltauv v) \gets \computeDeltaV()$\; \label{line:computedeltavshort}
$\tau \gets 0$\;

\lForEach{$v \in V$}{$\btildev v \gets 0$}

\While{$\tau < \omega$ and not \testStoppingCondition$(\btilde, \deltal, \deltau ,\omega,\fixtime)$}{\label{line:stopcondshort} 
	$\p = \samplePath()$\;
	\lForEach{$v \in \p$}{$\btildev v \gets \btildev v+1$}
    $\tau \gets \tau+1$\;
}

\lForEach{$v \in V$}{$\btildev v \gets \btildev v/\tau$}
\Return{$\btilde$}
\end{footnotesize}
\caption{our algorithm for approximating betweenness centrality.}
\label{alg:mainshort}
\end{algorithm2e}

The correctness of the algorithm follows from the following theorem, which is the base of our adaptive sampling, and which we prove in \Cref{sec:adaptive} (where we also define the functions $f$ and $g$).

\begin{thm} \label{thm:mainshort}
Let $\btildev v$ be the output of Algorithm~\ref{alg:mainshort}, and let $\stoptime$ be the number of samples at the end of the algorithm. Then, with probability $1-\delta$, the following conditions hold:
\begin{itemize}
    \item if $\stoptime=\omega$, $|\btildev v-\betv v| < \lambda$ for all $v$;
    \item if $\stoptime<\omega$, $ -f(\stoptime,\btildev v,\deltalv v,\omega) \leq \betv v - \btildev v \leq  g(\stoptime,\btildev v,\deltauv v,\omega)$ for all $v$.
\end{itemize}
\end{thm}
\begin{remark}
    This theorem says that, at the beginning of the algorithm, we know that, with probability $1-\delta$, one of the two conditions will hold when the algorithm stops, independently of the final value of $\stoptime$. This is essential to avoid the stochastic dependence that we discuss in \Cref{sec:adapshort}.
\end{remark}

In order to apply this theorem, we choose $\lambda$ such that our goal is reached if all centralities are known with error at most $\lambda$. Then, we choose the function \testStoppingCondition\ in a way that our goal is reached if the stopping condition is satisfied. This way, our algorithm is correct, both if $\stoptime=\omega$ and if $\stoptime<\omega$.
For example, if we want to compute all centralities with bounded absolute error, we simply choose $\lambda$ as the bound we want to achieve, and we plug the stopping condition $f,g\leq \lambda$ in the function \testStoppingCondition. 
Instead, if we want to compute an approximation of the $k$ most central nodes, we need a different definition of $f$ and $g$, which is provided in \Cref{sec:topkshort}.

To complete the description of this algorithm, we need to specify the following functions.
\begin{description}
\item[\computeDeltaV] The algorithm works for any choice of the $\deltalv v,\deltauv v$s, but a good choice yields better running times. We propose a heuristic way to choose them in \Cref{sec:deltav}.

\item[\samplePath] In order to sample a path between two random nodes $s$ and $t$, we use a balanced bidirectional BFS, which is defined in \Cref{sec:bbbfs}.
\end{description}

\section{Adaptive Sampling}
\label{sec:adapshort}

In this section, we highlight the main technical difficulty in the formalization of adaptive sampling, which previous works claiming analogous results did not address. Furthermore, we sketch the way we overcome this difficulty: our argument is quite general, and it could be easily adapted to formalize these claims.

As already said, the problem is the stochastic dependence between the time $\stoptime$ in which the algorithm terminates and the event $\good_\fixtime=$ ``at time $\fixtime$, the estimate is within the required distance from the true value'', since both $\stoptime$ and $\good_\fixtime$ are functions of the same random sample. Since it is typically possible to prove that $\Pr(\neg \good_\fixtime) \leq \delta$ for every fixed $\fixtime$, one may be tempted to argue that also $\Pr(\neg \good_{\stoptime}) \leq \delta$, by applying these inequalities at time $\stoptime$. However, this is not correct: indeed, if we have no assumptions on $\stoptime$, $\stoptime$ could even be defined as the smallest $\fixtime$ such that $\good_\fixtime$ does not hold!

More formally, if we want to link $\Pr(\neg \good_{\stoptime})$ to $\Pr(\neg \good_\fixtime)$, we have to use the law of total probability, that says that:
\begin{align}
    \Pr(\neg \good_{\stoptime}) 
    &= \sum_{\tau=1}^{\infty} \Pr(\neg \good_{\stoptime}\,|\, \stoptime=\fixtime)\Pr(\stoptime=\fixtime)
    \label{eq:first_tryshort}\\
    &= \Pr(\neg \good_{\stoptime}\,|\, \stoptime < \fixtime)\Pr(\stoptime < \fixtime) + \Pr(\neg \good_{\stoptime}\,|\, \stoptime\geq \fixtime)\Pr(\stoptime \geq \fixtime).
    \label{eq:second_tryshort}
\end{align}
Then, if we want to bound $\Pr(\neg \good_{\stoptime})$, we need to assume that 
\begin{equation}
    \Pr(\neg A_{\stoptime}\,|\, \stoptime=\fixtime) \leq \Pr(\neg A_\fixtime) \quad\text{or that} \quad
    \Pr(\neg A_{\tau}\,|\, \stoptime\geq \fixtime) \leq \Pr(\neg A_\fixtime),
    \label{eq:first_try_failshort}
\end{equation}
which would allow to bound \eqref{eq:first_tryshort} or \eqref{eq:second_tryshort} from above. The equations in \eqref{eq:first_try_failshort} are implicitly assumed to be true in previous works adopting adaptive sampling techniques. Unfortunately, because of the stochastic dependence, it is quite difficult to prove such inequalities, even if some approaches managed to overcome these difficulties \cite{pietracaprina_mining_2010}. 

For this reason, our proofs avoid dealing with such relations: in the proof of \Cref{thm:mainshort}, we fix a deterministic time $\omega$, we impose that $\stoptime\leq\omega$, and we apply the inequalities with $\fixtime=\omega$. Then, using martingale theory, we convert results that hold at time $\omega$ to results that hold at the stopping time $\stoptime$ (see \Cref{sec:adaptive}).

\section{Balanced Bidirectional BFS}
\label{sec:bfsshort}

A major improvement of our algorithm, with respect to previous counterparts, is that we sample shortest paths through a balanced bidirectional BFS, instead of a standard BFS. In this section, we describe this technique, and we bound its running time on realistic models of random graphs, with high probability.
The idea behind this technique is very simple: if we need to sample a uniformly random shortest path from $s$ to $t$, instead of performing a full BFS from $s$ until we reach $t$, we perform at the same time a BFS from $s$ and a BFS from $t$, until the two BFSs touch each other (if the graph is directed, we perform a ``forward'' BFS from $s$ and a ``backward'' BFS from $t$).

More formally, assume that we have visited up to level $l_s$ from $s$ and to level $l_t$ from $t$, let $\Gamma^{l_s}(s)$ be the set of nodes at distance $l_s$ from $s$, and similarly let $\Gamma^{l_t}(t)$ be the set of nodes at distance $l_t$ from $t$. If $\sum_{v \in \Gamma^{l_s}(s)} \deg(v) \leq \sum_{v \in \Gamma^{l_t}(t)} \deg(v)$, we process all nodes in $\Gamma^{l_s}(s)$, otherwise we process all nodes in $\Gamma^{l_t}(t)$ (since the time needed to process level $l_s$ is proportional to $\sum_{v \in \Gamma^{l_s}(s)} \deg(v)$, this choice minimizes the time needed to visit the next level). Assume that we are processing the node $v \in \Gamma^{l_s}(s)$ (the other case is analogous). For each neighbor $w$ of $v$ we do the following:
\begin{itemize}
\item if $w$ was never visited, we add $w$ to $\Gamma^{l_s+1}(s)$;
\item if $w$ was already visited in the BFS from $s$, we do not do anything;
\item if $w$ was visited in the BFS from $t$, we add the edge $(v,w)$ to the set $\Pi$ of candidate edges in the shortest path.
\end{itemize}

After we have processed a level, we stop if $\Gamma^{l_s}(s)$ or $\Gamma^{l_t}(t)$ is empty (in this case, $s$ and $t$ are not connected), or if $\Pi$ is not empty. In the latter case, we select an edge from $\Pi$, so that the probability of choosing the edge $(v,w)$ is proportional to $\sigma_{sv}\sigma_{wt}$ (we recall that $\sigma_{xy}$ is the number of shortest paths from $x$ to $y$, and it can be computed during the BFS as in \cite{Brandes2007}). Then, the path is selected by considering the concatenation of a random path from $s$ to $v$, the edge $(v,w)$, and a random path from $w$ to $t$. These random paths can be easily chosen by backtracking, as shown in \cite{Riondato2015} (since the number of paths might be exponential in the input size, in order to avoid pathological cases, we assume that we can perform arithmetic operations in $\O(1)$ time).

\subsection{Analysis on Random Graph}

In order to show the effectiveness of the balanced bidirectional BFS, we bound its running time in several models of random graphs: the Configuration Model (CM, \cite{Bollobas1980}), and Rank-1 Inhomogeneous Random Graph models (IRG, \cite[Chapter 3]{Hofstad2014}), such as the Chung-Lu model \cite{Chung2006}, the Norros-Reittu model \cite{Norros2006}, and the Generalized Random Graph \cite[Chapter 3]{Hofstad2014}. In these models, we fix the number $n$ of nodes, and we give a weight $\w u$ to each node. In the CM, we create edges by giving $\w u$ half-edges to each node $u$, and pairing these half-edges uniformly at random; in IRG we connect each pair of nodes $(u,v)$ independently with probability close to ${\w u \w v}/{\sum_{w \in V} \w w}$. With some technical assumptions discussed in \Cref{sec:bbbfs}, we prove the following theorem.
\begin{thm} \label{thm:bidirectionalshort}
Let $G$ be a graph generated through the aforementioned models. Then, for each fixed $\epsilon>0$, and for each pair of nodes $s,t$, \whp, the time needed to compute an $st$-shortest path through a bidirectional BFS is $\O(n^{\frac{1}{2}+\epsilon})$ if the degree distribution $\lambda$ has finite second moment, $\O(n^{\frac{4-\beta}{2}+\epsilon})$ if $\lambda$ is a power law distribution with $2<\beta<3$.
\end{thm}
\begin{proof}[Sketch of proof]
The idea of the proof is that the time needed by a bidirectional BFS is proportional to the number of visited edges, which is close to the sum of the degrees of the visited nodes, which are very close to their weights. Hence, we have to analyze the weights of the visited edges: for this reason, if $V'$ is a subset of $V$, we define the volume of $V'$ as $\w{V'}=\sum_{v \in V'} \w v$.

Our visit proceeds by ``levels'' in the BFS trees from $s$ and $t$: if we never process a level with total weight at least $n^{\frac{1}{2}+\epsilon}$, since the diameter is $\O(\log n)$, the volume of the set of processed vertices is $\O(n^{\frac{1}{2}+\epsilon}\log n)$, and the number of visited edges cannot be much bigger (for example, this happens if $s$ and $t$ are not connected). Otherwise, assume that, at some point, we process a level $l_s$ in the BFS from $s$ with total weight $n^{\frac{1}{2}+\epsilon}$: then, the corresponding level $l_t$ in the BFS from $t$ has also weight $n^{\frac{1}{2}+\epsilon}$ (otherwise, we would have expanded from $t$, because weights and degrees are strongly correlated). We use the ``birthday paradox'': levels $l_s+1$ in the BFS from $s$, and level $l_t+1$ in the BFS from $t$ are random sets of nodes with size close to $n^{\frac{1}{2}+\epsilon}$, and hence there is a node that is common to both, \whp. This means that the time needed by the bidirectional BFS is proportional to the volume of all levels in the BFS tree from $s$, until $l_s$, plus the volume of all levels in the BFS tree from $t$, until $l_t$ (note that we do not expand levels $l_s+1$ and $l_t+1$). All levels except the last have volume at most $n^{\frac{1}{2}+\epsilon}$, and there are $\O(\log n)$ such levels because the diameter is $\O(\log n)$: it only remains to estimate the volume of the last level. 

By definition of the models, the probability that a node $v$ with weight $\w v$ belongs to the last level is about $\frac{\w v \w{\G{l_s-1}s}}{M} \leq \w v n^{-\frac{1}{2}+\epsilon}$: hence, the expected volume of $\G{l_s}s$ is at most $\sum_{v \in V} \w v \Pr(v \in \G{l_s-1}s) \leq \sum_{v \in V} \w v^2 n^{-\frac{1}{2}+\epsilon}$. Through standard concentration inequalities, we prove that this random variable is concentrated: hence, we only need to compute this expected value. If the degree distribution has finite second moment, then $\sum_{v \in V} \w v^2=\O(n)$, concluding the proof. If the degree distribution is power law with $2<\beta<3$, then we have to consider separately nodes $v$ such that $\w v<n^{\frac{1}{2}}$ and such that $\w v>n^{\frac{1}{2}}$. In the first case, $\sum_{\w v<n^{\frac{1}{2}}} \w v^2 \approx \sum_{d=0}^{n^{\frac{1}{2}}} nd^2\lambda(d) \approx \sum_{d=0}^{n^{\frac{1}{2}}} nd^{2-\beta}\approx n^{1+\frac{3-\beta}{2}}$. In the second case, we prove that the volume of the set of nodes with weight bigger than $n^{\frac{1}{2}}$ is at most $n^{\frac{4-\beta}{2}}$. Hence, the total volume of $\G{l_s}s$ is at most $n^{-\frac{1}{2}+\epsilon}n^{1+\frac{3-\beta}{2}}+n^{\frac{4-\beta}{2}} \approx n^{\frac{4-\beta}{2}}$.

\end{proof}

\section{Computing the \texorpdfstring{$k$}{k} Most Central Nodes} \label{sec:topkshort}

Differently from previous works, our algorithm is more flexible, making it possible to compute the betweenness centrality of different nodes with different precision. This feature can be exploited if we only want to rank the nodes: for instance, if $v$ is much more central than all the other nodes, we do not need a very precise estimation on the centrality of $v$ to say that it is the top node. Following this idea, in this section we adapt our approach to the approximation of the ranking of the $k$ most central nodes: as far as we know, this is the first approach which computes the ranking without computing a $\lambda$-approximation of all betweenness centralities, allowing significant speedups. 
Clearly, we cannot expect our ranking to be always correct, otherwise the algorithm does not terminate if two of the $k$ most central nodes have the same centrality. For this reason, the user fixes a parameter $\lambda$, and, for each node $v$, the algorithm does one of the following:
\begin{itemize}
\item it provides the exact position of $v$ in the ranking;
\item it guarantees that $v$ is not in the top-$k$;
\item it provides a value $\btildev v$ such that $|\betv v-\btildev v|\leq\lambda$.
\end{itemize}

In other words, similarly to what is done in \cite{Riondato2015}, the algorithm provides a set of $k'\geq k$ nodes containing the top-$k$ nodes, and for each pair of nodes $v,w$ in this subset, either we can rank correctly $v$ and $w$, or $v$ and $w$ are almost even, that is, $|\betv v-\betv w|\leq 2\lambda$. In order to obtain this result, we plug into Algorithm~\ref{alg:mainshort} the aforementioned conditions in the function \testStoppingCondition\ (see Algorithm~\ref{alg:stopcondtopk} in the appendix).

Then, we have to adapt the function \computeDeltaV\ to optimize the $\deltalv v$s and the $\deltauv v$s to the new stopping condition: in other words, we have to choose the values of $\lambdalv v$ and $\lambdauv v$ that should be plugged into the function \computeDeltaV\ (we recall that the heuristic \computeDeltaV\ chooses the $\deltalv v$s so that we can guarantee as fast as possible that $\btildev v-\lambdalv v \leq \bet (v) \leq \btildev v+\lambdauv v$). To this purpose, we estimate the betweenness of all nodes with few samples and we sort all nodes according to these approximate values $\tilde{b}(v)$, obtaining $v_1,\dots,v_n$. The basic idea is that, for the first $k$ nodes, we set $\lambdauv {v_i}=\frac{\tilde{b}(v_{i-1})-\tilde{b}(v_i)}{2}$, and $\lambdalv{v_i}=\frac{\tilde{b}(v_{i})-\tilde{b}(v_{i+1})}{2}$ (the goal is to find confidence intervals that separate the betweenness of $v_i$ from the betweenness of $v_{i+1}$ and $v_{i-1}$). For nodes that are not in the top-$k$, we choose $\lambdalv v=1$ and $\lambdauv v=\tilde{b}(v_{k})-\lambdalv{v_k}-\tilde{b}(v_{i})$ (the goal is to prove that $v_i$ is not in the top-$k$). Finally, if $\tilde{b}(v_{i})-\tilde{b}(v_{i+1})$ is small, we simply set $\lambdalv{v_i}=\lambdauv{v_i}=\lambdalv{v_{i+1}}=\lambdauv{v_{i+1}}=\lambda$, because we do not know if $\betv{v_{i+1}}>\betv{v_i}$, or viceversa. 

\section{Experimental Results} \label{sec:experimentsshort}

In this section, we test the four variations of our algorithm on several real-world networks, in order to evaluate their performances. The platform for our tests is a server with 1515 GB RAM and 48 Intel(R) Xeon(R) CPU E7-8857 v2 cores at 3.00GHz, running Debian GNU Linux 8. The algorithms are implemented in C++, and they are compiled using gcc 5.3.1. The source code of our algorithm is available at \url{https://sites.google.com/a/imtlucca.it/borassi/publications}.

\subsubsection*{Comparison with the State of the Art}

The first experiment compares the performances of our algorithm \cad\ with the state of the art.
The first competitor is the \rk\ algorithm \cite{Riondato2015}, available in the open-source \textit{NetworKit}
framework~\cite{Staudt2014}. This algorithm uses the same estimator as our algorithm, but the stopping condition is different: it simply stops after sampling $k=\frac{c}{\epsilon^2}\left(\left\lfloor\log_2(\vd-2)\right\rfloor+1+\log\left(\frac{1}{\delta}\right)\right)$, and it uses a heuristic to upper bound the vertex diameter. Following suggestions by the author of the \textit{NetworKit} implementation, we set to $20$ the number of samples used in the latter heuristic \cite{elisabetta_personal}.

The second competitor is the \abra\ algorithm \cite{Riondato2016}, available at \url{http://matteo.rionda.to/software/ABRA-radebetw.
tbz2}. This algorithm samples pairs of nodes $(s,t)$, and it adds the fraction of $st$-paths passing from $v$ to the approximation of the betweenness of $v$, for each node $v$. The stopping condition is based on a key result in statistical learning theory, and there is a scheduler that decides when it should be tested. Following the suggestions by the authors, we use both the automatic scheduler \abraaut, which uses a heuristic approach to decide when the stopping condition should be tested, and the geometric scheduler \abrag, which tests the stopping condition after $(1.2)^ik$ iterations, for each integer $i$.

The test is performed on a dataset made by $15$ undirected and $15$ directed real-world networks, taken from the datasets SNAP (\url{snap.stanford.edu/}), LASAGNE (\url{piluc.dsi.unifi.it/lasagne}), and KONECT (\url{http://konect.uni-koblenz.de/networks/}). As in \cite{Riondato2016}, we have considered all values of $\lambda \in\{0.03, 0.025, 0.02, 0.015, 0.01, 0.005\}$, and $\delta=0.1$. All the algorithms have to provide an approximation $\tilde{\boldsymbol{b}}(v)$ of $\betv v$ for each $v$ such that $\Pr\left(\forall v, \left|\tilde{\boldsymbol{b}}(v)-\betv v\right| \leq \lambda\right) \geq 1-\delta$. In \Cref{fig:running_timeshort}, we report the time needed by the different algorithms on every graph for $\lambda=0.005$ (the behavior with different values of $\lambda$ is very similar). More detailed results are reported in \Cref{app:detailedresults}.

\begin{figure}
    \includegraphics[width=\textwidth]{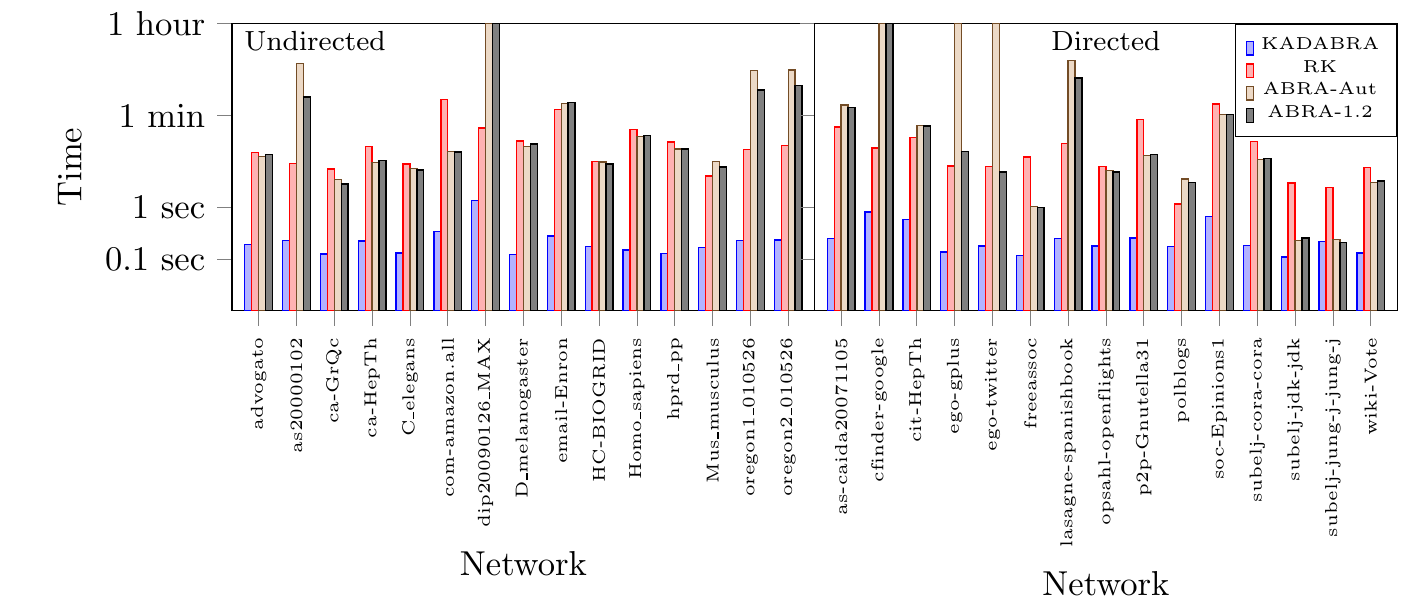}
    \caption{The time needed by the different algorithms, on all the graphs of our dataset.} \label{fig:running_timeshort}
\end{figure}

From the figure, we see that \cad\ is much faster than all the other algorithms, on all graphs: on average, our algorithm is about $100$ times faster than \rk\ in undirected graphs, and about $70$ times faster in directed graphs; it is also more than $1\,000$ times faster than \abra. The latter value is due to the fact that the \abra\ algorithm has large running times on few networks: in some cases, it did not even conclude its computation within one hour. The authors confirmed that this behavior might be due to some bugs in the code, which seems to affect it only on specific graphs: indeed, in most networks, the performances of \abra\ are better than those of the \rk\ algorithm (but, still, not better than \cad). 

\begin{figure}[t]
        \includegraphics[width=\textwidth]{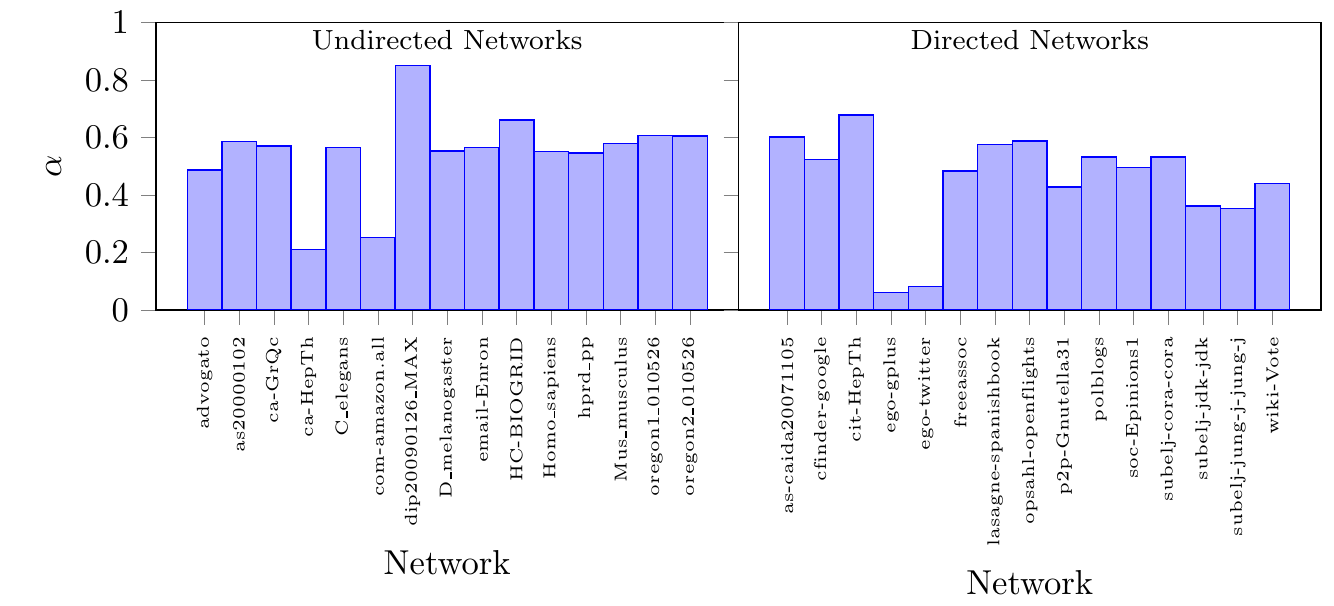}
    \caption{The exponent $\alpha$ such that the average number of edges visited during a bidirectional BFS is $n^\alpha$.} \label{fig:bidbfsshort}
\end{figure}

In order to explain these data, we take a closer look at the improvements obtained through the bidirectional BFS, by considering the average number of edges $m_{\avg}$ that the algorithm visits in order to sample a shortest path (for all our competitors, $m_{\avg}=m$, since they perform a full BFS). In \Cref{fig:bidbfsshort}, for each graph in our dataset, we plot $\alpha=\frac{\log(m_{\avg})}{\log(m)}$ (intuitively, this means that the average number of edges visited is $m^{\alpha}$).

The figure shows that, apart from few cases, the number of edges visited is close to $n^{\frac{1}{2}}$, confirming the results in \Cref{sec:bfsshort}. This means that, since many of our networks have approximately $10\,000$ edges, the bidirectional BFS is about $100$ times faster than the standard BFS. Finally, for each value of $\lambda$, we report in \Cref{fig:nsamplesshort} the number of samples needed by all the algorithms, averaged over all the graphs in the dataset.

\begin{figure}[t]
    \includegraphics[width=\textwidth]{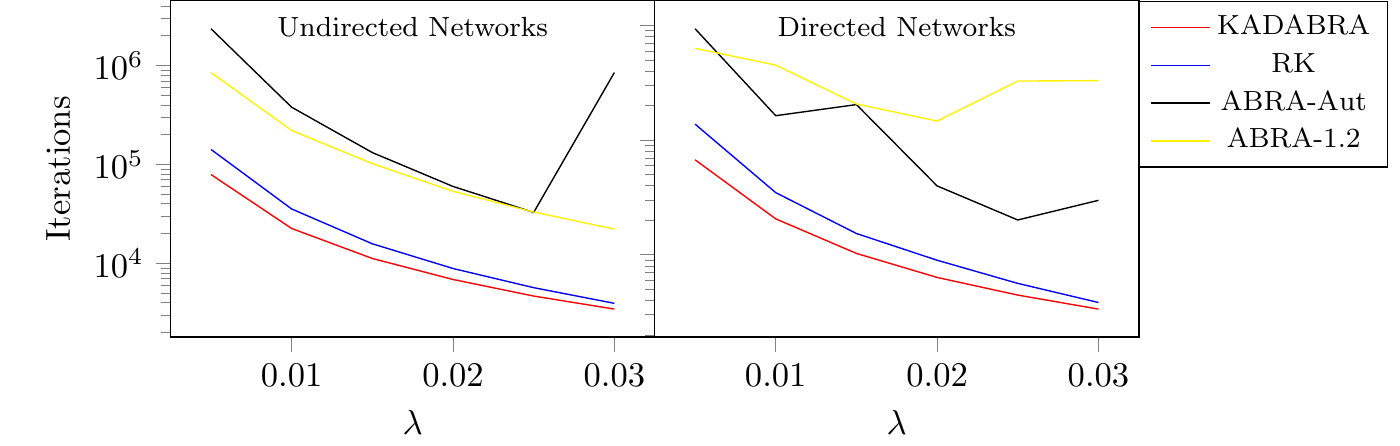}
    \caption{The average number of samples needed by the different algorithms.} \label{fig:nsamplesshort}
\end{figure}

From the figure, \cad\ needs to sample the smallest amount of shortest paths, and the average improvement over \rk\ grows when $\lambda$ tends to $0$, from a factor $1.14$ (resp., $1.14$) if $\lambda=0.03$, to a factor $1.79$ (resp., $2.05$) if $\lambda=0.005$ in the case of undirected (resp., directed) networks. Again, the behavior of \abra\ is highly influenced by the behavior on few networks, and as a consequence the average number of samples is higher. In any case, also in the graphs where \abra\ has good performances, \cad\ still needs a smaller number of samples.

\subsubsection*{Computing Top-\texorpdfstring{$k$}{k} Centralities}

In the second experiment, we let \cad compute the top-$k$ betweenness centralities of large graphs, 
which were unfeasible to handle with the previous algorithms.

The first set of graph is a series of temporal snapshots of the IMDB actor collaboration network, in which two actors are connected if they played together in a movie. The snapshots are taken every 5 years from 1940 to 2010, including a last snapshot in 2014, with $1\,797\,446$ nodes and $145\,760\,312$ edges. The graphs are extracted from the IMDB website (\url{http://www.imdb.com}), and they do not consider TV-series, awards-shows, documentaries, game-shows, news, realities and talk-shows, in accordance to what was done in \url{http://oracleofbacon.org}.

The other graph considered is the Wikipedia citation network, whose nodes are Wikipedia pages, and which contains an edge from page $p_1$ to page $p_2$ if the text of page $p_1$ contains a link to page $p_2$. The graph is extracted from DBPedia 3.7 (\url{http://wiki.dbpedia.org/}), and it consists of $4\,229\,697$ nodes and $102\,165\,832$ edges.

\begin{figure}
    \centering
    \includegraphics[width=\textwidth]{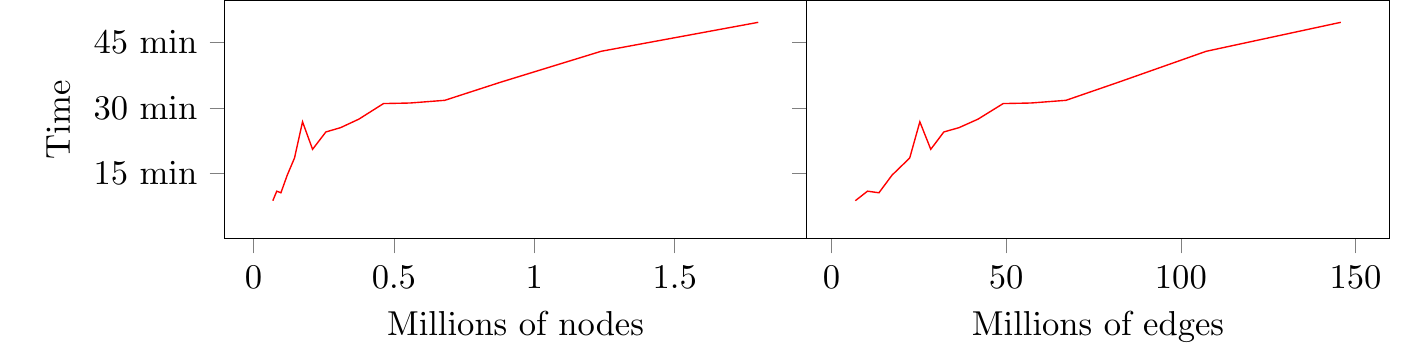}
    \caption{The total time of computation of \cad on increasing snapshots of the IMDB graph.}
    \label{fig:imdb_asymshort}
\end{figure}

We have run our algorithm with $\lambda=0.0002$ and $\delta=0.1$: as discussed in \Cref{sec:topkshort}, this means that either two nodes are ranked correctly, or their centrality is known with precision at most $\lambda$. As a consequence, if two nodes are not ranked correctly, the difference between their real betweenness is at most $2\lambda$. The full results are available in \Cref{sec:wiki_results}.

All the graphs were processed in less than one hour, apart from the Wikipedia graph, which was processed in approximately $1$ hour and $38$ minutes. In \Cref{fig:imdb_asymshort}, we plot the running times for the actor graphs: from the figure, it seems that the time needed by our algorithm scales slightly sublinearly with respect to the size of the graph. This result respects the results in \Cref{sec:bfsshort}, because the degrees in the actor collaboration network are power law distributed with exponent $\beta \approx 2.13$ (\url{http://konect.uni-koblenz.de/networks/actor-collaboration}). Finally, we observe that the ranking is quite precise: indeed, most of the times, there are very few nodes in the top-$5$ with the same ranking, and the ranking rarely contains significantly more than $10$ nodes.

\vfill

\subparagraph{Acknowledgements.} The authors would like to thank Matteo Riondato for several constructive comments on an earlier version of this work. We also thank Elisabetta Bergamini, Richard Lipton, and Sebastiano Vigna for helpful discussions and Holger Dell for his help with the experiments.

\newpage

\bibliography{full-version}

\appendix

\newpage

\section{Pseudocode}

\begin{algorithm2e}[H]
 \begin{footnotesize}
\LinesNumbered
\SetKwFunction{enqueue}{enqueue}
\SetKwFunction{extractMin}{extractMin}
\SetKwFunction{samplePath}{samplePath}
 \SetKwInOut{Input}{Input}
 \SetKwInOut{Output}{Output}
 
\Input{a graph $G=(V,E)$, and two values $\lambdalv v, \lambdauv v$ for each $v \in V$}
\Output{for each $v\in V$, two values $\deltalv v,\deltauv v$}
$\alpha \gets \frac{\omega}{100}$\;
$\epsilon \gets 0.0001$\;

\ForEach {$i \in [1,\alpha]$}{
	$\p = \samplePath()$\;
	\lForEach{$v \in \p$}{$\btildev v \gets \btildev v+1$}
}

\ForEach{$v \in V$}{$\btildev v \gets \btildev v/\alpha$\;
$c_L(v) \gets \frac{2\tilde{b}(v)\omega}{(\lambdalv v)^2}$\;
$c_U(v) \gets \frac{2\tilde{b}(v)\omega}{(\lambdauv v)^2}$\;
}

Binary search to find $C$ such that $\sum_{v \in V} \exp\left(-\frac{C}{c_L(v)}\right)+\exp\left(-\frac{C}{c_U(v)}\right)=\frac{\delta}{2}-\epsilon\delta$\;
\ForEach{$v \in V$}{
    $\deltalv v \gets \exp\left(-\frac{C}{c_L(v)}\right)+\frac{\epsilon\delta}{2n}$\; 
    $\deltauv v\gets \exp\left(-\frac{C}{c_U(v)}\right)+\frac{\epsilon\delta}{2n}$\;
}
\Return{$\boldsymbol{b}$}\;
\end{footnotesize}
\caption{the function \protect\computeDeltaV.}
\label{alg:choosedeltav}
\end{algorithm2e}

\begin{algorithm2e}[H]
 \begin{footnotesize}
\LinesNumbered
 \SetKwInOut{Input}{Input}
 \SetKwInOut{Output}{Output}
 
\Input{for each node $v$, the values of $\btildev v, \deltalv v,\deltauv v$, and the values of $\omega$ and $\fixtime$}
\Output{True if the algorithm should stop, False otherwise}
Sort nodes in decreasing order of $\btildev v$, obtaining $v_1,\dots,v_n$\;

\For {$i\in [1,\dots,k]$}{
	\If{$f(\btildev {v_{i}},\deltalv {v_{i}},\omega,\fixtime)>\lambda$ or $g(\btildev {v_{i}},\deltauv {v_{i}},\omega,\fixtime)>\lambda$}{
	    \If{$\btildev{v_{i-1}}-f(\btildev {v_{i-1}},\deltalv {v_{i-1}},\omega,\fixtime)<\btildev{v_i}+g(\btildev {v_{i}},\deltauv {v_{i}},\omega,\fixtime)$ or $\btildev{v_{i}}-f(\btildev {v_{i}},\deltalv {v_{i}},\omega,\fixtime)<\btildev{v_{i+1}}+g(\btildev {v_{i+1}},\deltauv {v_{i+1}},\omega,\fixtime)$}{
            \Return{False}\;
	    }
	}
}
\For {$i\in [k+1,\dots,n]$}{
	\If{$f(\btildev {v_{i}},\deltalv {v_{i}},\omega,\fixtime)>\lambda$ or $g(\btildev {v_{i}},\deltauv {v_{i}},\omega,\fixtime)>\lambda$}{
	    \If{$\btildev{v_{k}}-f(\btildev {v_{k}},\deltalv {v_{k}},\omega,\fixtime)<\btildev{v_i}+g(\btildev {v_{i}},\deltauv {v_{i}},\omega,\fixtime)$}{
            \Return{False}\;
	    }
	}
}

\Return{True}\;

\end{footnotesize}
\caption{the function \protect\testStoppingCondition\ to compute the top-$k$ nodes.}
\label{alg:stopcondtopk}
\end{algorithm2e}

\section{Concentration Inequalities} \label{app:concentration}

\begin{lem}[Hoeffding's inequality]
    \label{lem:hoeff}
    Let $\X_1,\dots,\X_k$ be independent random variables such that $a_i<\X_i<b_i$, and let $\X=\sum_{i=1}^k \X_i$. Then, 
    \[
        \Pr\left(|\X-\E[X]| \geq \lambda\right) \leq \exp\left\{-\frac{2\lambda^2}{\sum_{i=1}^k(b_i-a_i)^2}\right\}.
    \]
\end{lem}

\begin{remark}
    If we apply Hoeffding's inequality with $\X_i=\Xvp v\Path$, $\X=k\ba v = \sum_{i=1}^k \Xvp v\Path$, $a_i=0$, $b_i=1$, we obtain that $\P\left(\left|\ba v- \betv v\right|>\lambda\right)<2e^{-2k\lambda^2}$. 
    Then, if we fix $\delta=2e^{-2k\lambda^2}$, the error is $\lambda=\sqrt{\frac{\log(2/\delta)}{2k}}$, and the minimum $k$ needed to obtain an error $\lambda$ on the betweenness of a single node is $\frac{1}{2\lambda^2}\log(2/\delta)$.
\end{remark}

\begin{lem}[Chernoff bound (\cite{Chung2006})]
    \label{lem:cb}
    Let $\X_1,\dots,\X_k$ be independent random variables such that $\X_i\leq M$ for each $1\leq i\leq n$, and let $\X=\sum_{i=1}^k \X_i$. Then, 
    \[
        \Pr\left(\X \geq \E[\X] + \lambda \right) \leq     \exp\left\{- \frac{\lambda^2}{2(\sum_{i=1}^n \E[\X_i^2] + M \lambda /3)}\right\}.
    \]
\end{lem}

\begin{thm}[McDiarmid '98 (\cite{Chung2006})]
    \label{thm:mcdiarmid}
    Let $X$ be a martingale associated with a filter $\mathcal{F}$, satisfying 
    \begin{itemize}
        \item $\mathrm{Var}\left(X_{i}\middle|\ \mathcal{F}_{i}\right)\leq\sigma_{i}$
        for $1\leq i\leq\ell$, 
        \item $\left|X_{i}-X_{i-1}\right|\leq M$, for $1\leq i\leq\ell$.
    \end{itemize}
    Then, we have
    \[
        \Pr\left(X-\mathbb{E}\left(X\right)\geq\lambda\right)\leq\exp\left(-\frac{\lambda^{2}}{2\left(\sum_{i=1}^{\ell}\sigma_{i}^{2}+M\lambda/3\right)}\right).
    \]
\end{thm}
\section{Proof of Theorem~\ref{thm:mainshort}}
\label{sec:adaptive}

In our algorithm, we sample $\stoptime$ shortest paths $\p_i$, where $\stoptime$ is a random variable such that $\stoptime=\fixtime$ can be decided by looking at the first $\fixtime$ paths sampled (see Algorithm~\ref{alg:mainshort}). 
Furthermore, thanks to Eq. (3) in \cite{Riondato2015}, we assume that $\stoptime \leq \omega$ for some fixed $\omega \in \mathbb{R}^+$ such that, after $\omega$ steps, $\Pr(\forall v,|\btildev v-\betv v|\leq\lambda) \geq 1-\frac{\delta}{2}$.
When the algorithm stops, our estimate of the betweenness is $\btilde(v):=\frac{1}{\stoptime}\sum_{i=1}^{\stoptime} \X_i(v)$, where $\X_i(v)$ is $1$ if $v$ belongs to $\p_i$, $0$ otherwise.

To estimate the error, we use the following theorem.

\begin{thm}
\label{thm:azuma-borassi}
For each node $v$ and for every fixed real numbers $\deltal$, $\deltau$, it holds
\begin{align*}
    \Pr\left(\betv v\leq\btildev v-f\left(\btilde(v),\deltal, \omega, \stoptime\right)\right) & \leq\deltal\quad\mbox{and}\\
    \Pr\left(\betv v\geq\btildev v+g\left(\btildev v,\deltau, \omega, \stoptime\right)\right) & \leq\deltau,
\end{align*}
where 
\begin{align}
    f\left(\btildev v,\deltal, \omega, \stoptime\right)
        &=\frac{1}{\stoptime} \logdl \left(\frac{1}{3}-\frac{\omega}{\stoptime}+\sqrt{\left(\frac{1}{3}-\frac{\omega}{\stoptime}\right)^{2}+\frac{2\btildev v\omega}{\logdl}}\right) \quad \text{and}
        \label{eq:thm_first}\\
    g\left(\btildev v,\deltau, \omega, \stoptime\right)
        &=\frac{1}{\stoptime} \logdu 
        \left(\frac{1}{3}+\frac{\omega}{\stoptime}
        +\sqrt{\left(\frac{1}{3}+\frac{\omega}{\stoptime}\right)^{2}+\frac{2\btildev v\omega}{\logdu}}\right).
        \label{eq:thm_second}
\end{align}
\end{thm}
We prove Theorem \ref{thm:azuma-borassi} in Section \ref{sec:azuma-borassi}. In the rest of this section, we show how this theorem implies \Cref{thm:mainshort}. To simplify notation, we often omit the arguments of the function $f$ and $g$.

\begin{proof}[Proof of \Cref{thm:mainshort}]
 
Let $\boldsymbol{E}_1$ be the event $(\stoptime=\omega \wedge \exists v \in V, |\btildev v-\betv v| > \lambda)$, and let $\boldsymbol{E}_2$ be the event $(\stoptime<\omega \wedge (\exists v \in V, -f \geq \betv v - \btildev v \vee \betv v - \btildev v \geq g))$. Let us also denote $\btildekv \fixtime{v}=\frac{1}{\fixtime}\sum_{i=1}^{\fixtime} \X_i(v)$ (note that $\btildekv \stoptime v=\btildev v$).

By our choice of $\omega$ and Eq. (3) in \cite{Riondato2015},
\[
    \Pr(\boldsymbol{E}_1) \leq \Pr(\exists v \in V,|\btildekv \omega v-\betv v| > \lambda)\leq \frac{\delta}{2}
\]
where $\btildekv{\omega}v$ is the approximate betweenness of $v$ after $\omega$ samples. 
Furthermore, by \Cref{thm:azuma-borassi}, 
\begin{align*}
    \Pr(\boldsymbol{E}_2) &\leq \sum_{v \in V}\Pr(\stoptime<\omega \wedge -f \geq \betv v - \btildev v)+\Pr(\stoptime<\omega \wedge \betv v - \btildev v \leq g)\\ 
    &\leq \sum_{v \in V} \deltalv v+\deltauv v \leq \frac{\delta}{2}.
\end{align*}
By a union bound, $\Pr(\boldsymbol{E}_1 \vee \boldsymbol{E}_2) \leq \Pr(\boldsymbol{E}_1)+\Pr(\boldsymbol{E}_1) \leq \delta$, concluding the proof of \Cref{thm:mainshort}.
\end{proof}

\subsection{Proof of Theorem \ref{thm:azuma-borassi}}
\label{sec:azuma-borassi}

Since this theorem deals with a single node $v$, let us simply write $\bet = \betv v, \btilde = \btildev v, \X_i=\X_i(v)$. 
Let us consider $\Y^\tau=\sum_{i=1}^\tau\left(\X_i-\bet\right)$ (we recall that $\X_i=1$ if $v$ is in the $i$-th path sampled, $\X_i=0$ otherwise). 
Clearly, $\Y^\tau$ is a martingale, and $\stoptime$ is a stopping time for $\Y^\tau$: this means that also $\Z^\tau=\Y^{\min(\stoptime,\tau)}$ is a martingale. 

Let us apply Theorem \ref{thm:mcdiarmid} to the martingales $\Z$ and $-\Z$: 
for each fixed $\lambdal,\lambdau>0$ we have 
\begin{align}
    \Pr\left(\Z^{\omega}\geq\lambdal\right) 
        & =\Pr\left(\stoptime\btilde-\stoptime \bet\geq\lambdal\right)\leq\exp\left(-\frac{\lambdal^{2}}{2\left(\omega \bet+\lambdal/3\right)}\right)=\deltal\quad\mbox{and}
        \label{eq:azuma_apply}\\
    \Pr\left(-\Z^{\omega}\geq\lambdau\right) 
        & =\Pr\left(\stoptime\btilde-\stoptime \bet\leq-\lambdau\right)\leq\exp\left(-\frac{\lambdau^{2}}{2\left(\omega \bet+\lambdau/3\right)}\right)=\deltau.      \label{eq:azuma_apply_second}
\end{align}
We now show how to prove \eqref{eq:thm_first} from \eqref{eq:azuma_apply}. 
The way to derive \eqref{eq:thm_second} from \eqref{eq:azuma_apply_second} is analogous. 

If we express $\lambdal$ as a function of $\deltal$ we get
\[
    \lambdal^{2}=2\logdl\left(\omega \bet+\frac{\lambdal}{3}\right)\iff\lambdal^{2}-\frac{2}{3}\lambdal\logdl-2\omega \bet\logdl=0,
\]
which implies that 
\[
    \lambdal=\frac{1}{3}\logdl\pm\sqrt{\frac{1}{9}\left(\logdl\right)^{2}+2\omega \bet\logdl}.
\]
Since \eqref{eq:azuma_apply} holds for any positive value $\lambdal$, it also holds for the value corresponding to the positive solution of this equation, that is,
\[
    \lambdal=\frac{1}{3}\logdl+\sqrt{\frac{1}{9}\left(\logdl\right)^{2}+2\omega \bet\logdl}.
\]
Plugging this value into \eqref{eq:azuma_apply}, we obtain
\begin{equation}
    \Pr\left(\stoptime\btilde-\stoptime \bet\geq\frac{1}{3}\logdl+\sqrt{\frac{1}{9}\left(\logdl\right)^{2}+2\omega \bet\logdl}\right)\leq\deltal.
    \label{eq:first_azume}
\end{equation}
By assuming $\btilde-\bet\geq \frac 1{3\stoptime} \log (\frac{1}{\deltal})$, the event in \eqref{eq:first_azume} can be rewritten as
\[
    \left(\stoptime \bet\right)^{2}-2\bet\left(\stoptime^{2}\btilde+\omega\logdl-\frac{1}{3}\stoptime\logdl\right)-\frac{2}{3}\logdl\stoptime\btilde+\left(\stoptime\btilde\right)^{2}\geq0.
\]
By solving the previous quadratic equation w.r.t. $\bet$ we get
\begin{equation*}
    \bet \leq
    \btilde+\logdl \left( 
        \frac{\omega}{\stoptime^{2}}
        -\frac{1}{3\stoptime}
        -\sqrt{
        \left(\frac{\btilde}{\logdl}+\frac{\omega}{\stoptime^{2}}-\frac{1}{3\stoptime}\right)^{2}
        -\left(\frac{\btilde}{\logdl}\right)^{2}+\frac{2}{3\stoptime}\frac{\btilde}{\logdl}}
    \right),
\end{equation*}
where we only considered the solution which upper bounds $\bet$, since we assumed $\btilde-\bet\geq \frac 1{3\tau} \log (\frac{1}{\deltal})$.
After simplifying the terms under the square root in the previous expression, we get 
\begin{equation*}
    \bet \leq
    \btilde+\logdl \left( 
    \frac{\omega}{\stoptime^{2}}-\frac{1}{3\stoptime}-\sqrt{\left(\frac{\omega}{\stoptime^{2}}-\frac{1}{3\stoptime}\right)^{2}+\frac{2\btilde\omega}{\stoptime^2\logdl}}
    \right),
\end{equation*}
which means that
\[
    \Pr\left(\bet\leq\btilde-f\left(\btilde, \deltal, \omega, \stoptime\right)\right)\leq\deltal,
\]
concluding the proof.

\section{How to Choose \texorpdfstring{$\deltalv v, \deltauv v$}{delta(v)}} 
\label{sec:deltav}

In \Cref{sec:adaptive}, we proved that our algorithm works for any choice of the values $\deltalv v, \deltauv v$. In this section, we show how we can heuristically compute such values, in order to obtain the best performances. 

For each node $v$, let $\lambdalv v, \lambdauv v$ be the lower and the upper maximum error that we want to obtain on the betweenness of $v$: if we simply want all errors to be smaller than $\lambda$, we choose $\lambdalv v, \lambdauv v = \lambda$, but for other purposes different values might be needed. We want to minimize the time $\tau$ such that the approximation of the betweenness at time $\tau$ is in the confidence interval required. In formula, we want to minimize
\begin{equation} \label{eq:conddeltav}
    \min\left\{\tau \in \mathbb{N}:\forall v \in V, \left(f\left(\btildekv{\tau}v,\deltalv v,\omega,\fixtime\right) \leq \lambdalv v \wedge g\left(\btildekv{\tau}v,\deltauv v,\omega,\fixtime\right) \leq \lambdauv v\right)\right\}
\end{equation}

where $\btildekv{\tau}v$ is the approximation of $\betv{v}$ obtained at time $\tau$, and
\begin{align*}
    f\left(\tau,\btilde_\tau,\deltal, \omega\right)&=\frac{1}{\tau} \logdl \left(\frac{1}{3}-\frac{\omega}{\tau}+\sqrt{\left(\frac{1}{3}-\frac{\omega}{\tau}\right)^{2}+\frac{2\btilde_\tau\omega}{\logdl}}\right) \quad \text{and}\\
    g\left(\tau,\btilde_\tau,\deltau, \omega\right)&=\frac{1}{\tau} \logdu \left(\frac{1}{3}+\frac{\omega}{\tau}+\sqrt{\left(\frac{1}{3}+\frac{\omega}{\tau}\right)^{2}+\frac{2\btilde_\tau\omega}{\logdu}}\right).
\end{align*}

The goal of this section is to provide deterministic values of $\deltalv v,\deltauv v$ that minimize the value in \eqref{eq:conddeltav}, and such that $\sum_{v \in V} \deltalv v+ \deltauv v<\frac{\delta}{2}$. To obtain our estimate, we replace $\btildekv \tau v$ with an approximation $\tilde{b}(v)$, that we compute by sampling $\alpha$ paths, before starting the algorithm (in our code, $\alpha=\frac{\omega}{100}$). Furthermore, we consider a simplified version of \eqref{eq:conddeltav}: in most cases, $\lambdal$ is much smaller than all other quantities in play, and since $\omega$ is proportional to $\frac{1}{\lambdal^2}$, we can safely assume $f(\tau,\tilde{b}(v),\deltalv v, \omega) \approx \sqrt{\frac{2\tilde{b}(v)\omega}{\tau^2}\logdl}$ and $g(\tau,\tilde{b}(v),\deltauv v, \omega) \approx \sqrt{\frac{2\tilde{b}(v)\omega}{\tau^2}\logdu}$. Hence, in place of the value in \eqref{eq:conddeltav}, our heuristic tries to minimize 
\[
    \min\left\{\tau \in \mathbb{N}:\forall v \in V, \sqrt{\frac{2\tilde{b}(v)\omega}{\tau^2}\logdlv v} \leq \lambdalv v \wedge \sqrt{\frac{2\tilde{b}(v)\omega}{\tau^2}\logduv v} \leq \lambdauv v\right\}.
\]
Solving with respect to $\tau$, we are trying to minimize
\[
    \max_{v \in V} \left(\max\left(\sqrt{\frac{2\tilde{b}(v)\omega}{\left(\lambdalv v\right)^2}\logdlv v}, \sqrt{\frac{2\tilde{b}(v)\omega}{\left(\lambdauv v\right)^2}\logduv v}\right)\right).
\]
which is the same as minimizing $\max_{v \in V} \max \left(c_L(v)\logdlv v,c_U(v)\logduv v\right)$ for some constants $c_L(v), c_U(v)$, conditioned on $\sum_{v \in V} \deltalv v+ \deltauv v<\frac{\delta}{2}$. 
We claim that, among the possible choices of $\deltalv v,\deltauv v$, the best choice makes all the terms in the maximum equal: otherwise, if two terms were different, we would be able to slightly increase and decrease the corresponding values, in order to decrease the maximum. 
This means that, for some constant $C$, for each $v$, $c_L(v)\logdlv v = c_U(v)\logdlv v=C$, that is, $\deltalv v=\exp(-\frac{C}{c_L(v)})$, $\deltauv v=\exp(-\frac{C}{c_U(v)})$. 
In order to find the largest constant $C$ such that $\sum_{v \in V} \deltalv v+ \deltauv v\leq \frac{\delta}{2}$, we use a binary search procedure on all possible constants $C$.

Finally, if $c_L(v)=0$ or $c_U(v)=0$, this procedure chooses $\deltalv v=0$: to avoid this problem, we impose $\sum_{v \in V} \deltalv v+ \deltauv v\leq \frac{\delta}{2}-\epsilon\delta$, and we add $\frac{\epsilon\delta}{2n}$ to all the $\deltalv v$s and all the $\deltauv v$s (in our code, we choose $\epsilon=0.001$). 
The pseudocode of the algorithm is available in Algorithm~\ref{alg:choosedeltav}.

\section{Balanced Bidirectional BFS on Random Graphs} \label{sec:bbbfs}

In this appendix, we formally prove that the bidirectional BFS is efficient in several models of random graphs: the Configuration Model (CM, \cite{Bollobas1980}), and Rank-1 Inhomogeneous Random Graph models (IRG, \cite[Chapter 3]{Hofstad2014}), such as the Chung-Lu model \cite{Chung2006}, the Norros-Reittu model \cite{Norros2006}, and the Generalized Random Graph \cite[Chapter 3]{Hofstad2014}. All these models are defined by fixing the number $n$ of nodes and $n$ weights $\w v$, and by creating edges at random, in a way that node $v$ gets degree close to $\w v$.

More formally, the edges are generated as follows:
\begin{itemize}
\item in the CM, each node is associated to $\w v$ half-edges, or stubs; edges are created by randomly pairing these $M=\sum_{v \in V} \w v$ stubs (we assume the number of stubs to be even, by adding a stub to a random node if necessary).
\item in IRG, an edge between a node $v$ and a node $w$ exists with probability $f(\frac{\w v\w w}{M})$, where $M=\sum_{v \in V} \w v$, and the existence of different edges is independent. Different choices of the function $f$ create different models.
\begin{itemize}
\item In general, we assume that $f$ satisfies the following conditions:
\begin{itemize}
\item $f$ is derivable at least twice in $0$;
\item $f$ is increasing;
\item $f'(0)=1$;
\end{itemize}
\item in the Chung-Lu model, $f(x)=\min(x,1)$;
\item in the Norros-Reittu model, $f(x)=1-e^{-x}$;
\item in the Generalized Random Graph model, $f(x)=\frac{x}{1+x}$.
\end{itemize}
\end{itemize}

It remains to define how we choose the weights $\w v$, when the number of nodes $n$ tends to infinity. In the line of previous works \cite{Norros2006,Fernholz2007,Hofstad2014}, we consider a sequence of graphs $G_i$, whose number of nodes $n_i$ tends to infinity, and whose degree distribution $\lambda_i$ satisfy the following:
\begin{enumerate}
    \item there is a probability distribution $\lambda$ such that the $\lambda_i$s tend to $\lambda$ in distribution;
    \item $M_1(\lambda_i)$ tends to $M_1(\lambda)<\infty$, where $M_1(\lambda)$ is the first moment of $\lambda$;
    \item one of the following two conditions hold:
    \begin{enumerate}
        \item $M_2(\lambda_i)$ tends to $M_2(\lambda)<\infty$, where $M_2(\lambda)$ is the second moment of $\lambda$;
        \item $\lambda$ is a power law distribution with $2<\beta<3$, and there is a global constant $C$ such that, for each $d$, $\Pr(\lambda_i\geq d)\leq \frac{C}{d^{\beta-1}}$. \label{cond:powerlaw}
    \end{enumerate}
\end{enumerate}

For example, these assumptions are satisfied with probability $1$ if we choose the degrees independently, according to a distribution $\lambda$ with finite mean \cite[Section 6.1,7.2]{Hofstad2014}.

\begin{remark}
Note that an aspect often neglected in previous work when it comes to computing shortest paths is the fact that the number of shortest paths between a pair of nodes may be exponential, thus requiring to work with a linear number of bits. While real-world complex networks are typically sparse with logarithmic diameter, in order to avoid such issue it is sufficient to assume that addition and comparison require constant time. 
\end{remark}
\begin{remark}
These assumptions cover the Erd\"os-Renyi random graph with constant average degree, and all power law distributions with $\beta>2$ (because, if $\beta>3$, then $M_2(\lambda)$ is finite).
\end{remark}
\begin{remark}
Assumption \ref{cond:powerlaw} seems less natural than the other assumptions. However, it is necessary to exclude pathological cases: for example, assume that $G_i$ has $n-2$ nodes chosen according to a power law distribution, and $2$ nodes $u,v$ with weight $n^{1-\epsilon}$. All assumption except \ref{cond:powerlaw} are satisfied, but the bidirectional BFS is not efficient, because if $s$ is a neighbor of $u$ with degree $1$, and $t$ is a neighbor of $v$ with degree $1$, then a bidirectional BFS from $s$ and $t$ needs to visit all neighbors of $u$ or all neighbors of $v$, and the time needed is $\Omega(n^{1-\epsilon})$. 
\end{remark}

We say that a random graph has a property $\pi$ asymptotically almost surely (\aas) if $\Pr(\pi(G_i))$ tends to $1$ when $n$ tends to infinity. We say that a random graph has a property $\pi$ with high probability (\whp) if $\frac{\Pr(\pi(G_i))}{n_i^k}$ tends to $0$ for each $k>0$.

Before proving the main theorem, we need two more definitions and a technical assumption.

\begin{definition}
In the CM, let $\wres=\frac{M_2(\lambda)}{M_1(\lambda)}-1$. In IRG, let $\wres=\frac{M_2(\lambda)}{M_1(\lambda)}$ (if $\lambda$ is a power law distribution with $2<\beta<3$, we simply define $\wres=+\infty$).
\end{definition}
\begin{definition}
Given a set $V' \subseteq V$, the volume of $V'$ is $\w{V'}=\sum_{v \in V'} \w v$. Furthermore, if $V'=\G ds$, we abbreviate $\w{\G ds}$ with $\r ls$.
\end{definition}

The value $\wres$ is closely related to $\frac{\r{l+1}s}{\r ls}$: informally, the expected value of this fraction is $\wres$. For this reason, if $\wres<1$, then the size of neighbors tends to decrease, and all connected components have $\O(\log n)$ nodes. Conversely, if $\wres>1$, then the size of neighbors tends to increase, and there is a \emph{giant component} of size $\Theta(n)$ (for a proof of these facts, see \cite[Section 2.3 and Chapter 4]{Hofstad2014}). Our last assumption is that $\wres>1$, in order to ensure the existence of the giant component.

Under these assumptions, we prove Theorem \ref{thm:bidirectionalshort}, following the sketch in \Cref{sec:bfsshort}. We start by linking the degrees and the weights of nodes.

\begin{lem} \label{lem:degsweight}
For each node $v$, $\w v n^{-\epsilon} \leq \deg(v) \leq \w v n^{\epsilon}$ \whp.
\end{lem}
\begin{proof}
We use \cite[Lemmas 32 and 37]{Borassi2016}\footnote{This paper uses a further assumption on IRG, but the proofs of Lemmas 32 and 39 do not rely on this assumption.}: these lemmas imply that, for each $\epsilon>0$, if $\w v>n^\epsilon$, $(1-\epsilon) \w v \leq \deg(v) \leq (1+\epsilon) \w v$ \whp. We have to handle the case where $\w v<n^\epsilon$: one of the two inequalities is empty, while for the other inequality we observe that, if we decrease the weight of $v$, the degree of $v$ can only decrease. Hence, if $\w v<n^\epsilon$, $\deg(v)<(1+\epsilon)n^{\epsilon}$, and the result follows by changing the value of $\epsilon$.
\end{proof}

Following the intuitive proof, we have linked the number of visited edges with their weights. Let us define an abbreviation for the volume of the nodes at distance $l$ from $s$.

\begin{definition}
We denote by $\r ls$ the volume of nodes at distance exactly $l$ from $s$. In the CM, we denote by $\R ls$ the set of stubs at distance $l$ from $s$.
\end{definition}

Now, we need to show that, if $\r {l_s}s,\r{l_t}t>n^{\frac{1}{2}+\epsilon}$, then $d(s,t)\leq l_s+l_t+2$ \whp.

\begin{lem} \label{lem:touchbid}
Assume that $\r {l_s}s>n^{\frac{1}{2}+\epsilon}$, $\r {l_t}{t}>n^{\frac{1}{2}+\epsilon}$, and $\r{l_s-1}s,\r{l_t-1}t<(1-\epsilon)n^{\frac{1}{2}+\epsilon}$. Then, $d(s,t) \leq l_s+l_t+2$.
\end{lem}
\begin{proof}
Let us assume that we know the structure of $\N {l_s}s$ and $\N{l_t}t$, that is, for each possible structure $S$ of the subgraph induced by all nodes at distance $l_s$ from $s$ and distance $l_t$ from $t$, let $E_S$ be the event that $\N{l_s}s$ and $\N{l_t}t$ are exactly $S$. If we prove that $\P(d(s,t) \leq l+l'+2 | E_{S})<\epsilon$, then $\P(\r {l+1}s>\r ls)=\sum_{S} \P(\r {l+1}s>\r ls| E_{S})\P(E_{S})<\sum_{S} \epsilon \P(E_{S})=\epsilon$. First of all, if $S$ is such that the two neighborhoods touch each other, $\P(d(s,t) \leq l+l'+2 | E_{S})=0<\epsilon$. Otherwise, we consider separately the CM and IRG.

In the CM, conditioned on $E_S$, the stubs that are paired with stubs in $\R{l_s}s$ are a random subset of the set of stubs that are not paired in $S$. This random subset has size at least $\epsilon n^{\frac{1}{2}+\epsilon} \geq n^{\frac{1+\epsilon}{2}}$ (because $\epsilon$ is a fixed constant, and $n$ tends to infinity). Since the total number of stubs is $\O(n)$, and since the number of stubs in $\R{l_t}t$ is at least $\epsilon n^{\frac{1+\epsilon}{2}}$, one of the stubs in $\R{l_t}t$ is paired with a stub in $\r{l_s}s$ \whp, and $d(s,t) \leq l_s+l_t+1$. 

In IRG, the probability that a node $v$ is not connected to any node in $\G{l_s}s$ is at most $\prod_{w \in \G {l_s}s} (1-f(\frac{\w v \w w}{M})) =\prod_{w \in \G {l_s}s}(1-\Omega(\frac{\w w}{M}))=\exp({-\sum_{w \in \G {l_s}s}\Omega(\frac{\w w}{M})})=\exp({-\Omega(\frac{\r {l_s}s}{M})})=1-\Omega(\frac{\r {l_s}s}{M})=1-\Omega(n^{-\frac{1}{2}+\epsilon})$. This means that $v$ belongs to $\G {l_s+1}s$ with probability $\Omega(n^{-\frac{1}{2}+\epsilon})$, and similarly it belongs to $\G {{l_t}+1}t$ with probability $\Omega(n^{-\frac{1}{2}+\epsilon})$. Since the two events are independent, the probability that $v$ belongs to both is $\Omega(n^{-1+2\epsilon})$. Since, for each node $v$, the events that $v$ belongs to $\G {l_s+1}s \cap \G {l_t+1}t$ are independent, by a straightforward application of Hoeffding's inequality, \whp, there is a node $v$ that belongs to $\G{l_s+1}s \cap \G {l_t+1}t$, and $d(s,t) \leq l_s+l_t+2$ \whp, concluding the proof.
\end{proof}

The next ingredient is used to bound the first integers $l_s,l_t$ such that $\r{l_s}s,\r{l_t}t>n^{\frac{1}{2}+\epsilon}$.

\begin{thm}[Theorem 5.1 in \cite{Fernholz2007} for the CM, Theorem 14.8 in \cite{Bollobas2007} for IRG (see also \cite{Hofstad2014,Borassi2016})] \label{thm:diameter}
The diameter of a graph generated through the aforementioned models is $\O(\log n)$.
\end{thm}

The last ingredient of our proof is an upper bound on the size of $\r {l_s}s$ and $\r{l_t}t$.

\begin{lem} \label{lem:lastlevel}
With high probability, for each $s \in V$ and for each $l$ such that $\sum_{i=0}^l \r ls<n^{\frac{1}{2}+\epsilon}$, $\r {l+1}s<n^{\frac{1}{2}+3\epsilon}$ if $\lambda$ has finite second moment, $\r {l+1}s<n^{\frac{4-\beta}{2}+3\epsilon}$ if $\lambda$ is power law with $2<\beta<3$.
\end{lem}
\begin{proof}
We consider separately nodes with weight at most $n^{\frac{1}{2}-2\epsilon}$ from nodes with bigger weights: in the former case, we bound the number of such nodes that are in $\R{l+1}s$, while in the latter case we bound the total number of nodes with weight at least $n^{\frac{1}{2}-2\epsilon}$. Let us start with nodes with the latter case.

\textbf{Claim:} for each $\epsilon$, $\sum_{\w v \geq n^{\frac{1}{2}-\epsilon}} \w v$ is smaller than $n^{\frac{1}{2}+3\epsilon}$ if $\lambda$ has finite second moment, and it is smaller than $n^{\frac{4-\beta}{2}+3\epsilon}$ if $\lambda$ is power law with $2<\beta<3$.

\begin{proof}[Proof of claim]
If $\lambda$ has finite second moment, by Chebyshev inequality, for each $\alpha$, 
\[
\P\left(\lambda_i>n^{\frac{1}{2}+\alpha}\right)\leq \frac{\var(\lambda_i)}{n^{1+2\alpha}} \leq \frac{M_2(\lambda_i)}{n^{1+2\alpha}} = \O\left(\frac{M_2(\lambda)}{n^{1+2\alpha}}\right) = \O\left(n^{-1-2\alpha}\right).
\]

For $\alpha=\epsilon$, this means that no node has weight bigger than $n^{\frac{1}{2}+\epsilon}$, and for $\alpha=-\epsilon$, this means that the number of nodes with weight bigger than $n^{\frac{1}{2}-\epsilon}$ is at most $n^{2\epsilon}$. We conclude that $\sum_{\w v \geq n^{\frac{1}{2}-\epsilon}} \w v \leq \sum_{\w v \geq n^{\frac{1}{2}-\epsilon}} n^{\frac{1}{2}+\epsilon} \leq n^{\frac{1}{2}+3\epsilon}$.

If $\lambda$ is power law with $2<\beta<3$, by Assumption~\ref{cond:powerlaw} the number of nodes with weight at least $d$ is at most $Cnd^{-\beta+1}$. Consequently, using Abel's summation technique,
\begin{align*}
    \sum_{\w v \geq n^{\frac{1}{2}-\epsilon}} \w v &= \sum_{d=\w v}^{+\infty} d|\{v:\w v = d\}| \\
    &= \sum_{d=n^{\frac{1}{2}-\epsilon}}^{+\infty} d(|\{v:\w v \geq d\}|-|\{v:\w v \geq d+1\}|) \\
    &= \sum_{d=n^{\frac{1}{2}-\epsilon}}^{+\infty} d|\{v:\w v \geq d\}|-\sum_{d=n^{\frac{1}{2}-\epsilon}+1}^{+\infty}(d-1)|\{v:\w v \geq d\}| \\
    &= n^{\frac{1}{2}-\epsilon}|\{v:\w v \geq n^{\frac{1}{2}-\epsilon}\}|+\sum_{d=n^{\frac{1}{2}-\epsilon}+1}^{+\infty} |\{v:\w v \geq d\}| \\
    &\leq Cn^{\frac{1}{2}-\epsilon}n^{1-\left(\frac{1}{2}-\epsilon\right)(\beta-1)}+\sum_{d=n^{\frac{1}{2}-\epsilon}+1}^{+\infty}Cnd^{-\beta+1} \\
    &=\O\left(n^{\frac{4-\beta}{2}+\epsilon\beta}+n^{1-\left(\frac{1}{2}-\epsilon\right)(\beta-2)}\right)=\O\left(n^{\frac{4-\beta}{2}+\epsilon\beta}\right).
\end{align*}
\end{proof}

By this claim, $\sum_{v \in \G{l+1}s,\w v \geq n^{\frac{1}{2}-2\epsilon}} \w v$ is smaller than $n^{\frac{1}{2}+6\epsilon}$ if $\lambda$ has finite second moment, and it is smaller than $n^{\frac{4-\beta}{2}+6\epsilon}$ if $\lambda$ is power law with $2<\beta<3$. To conclude the proof, we only have to bound $\sum_{v \in \G{l+1}s,\w v < n^{\frac{1}{2}-2\epsilon}} \w v$.

\textbf{Claim:} with high probability, $\sum_{v \in \G{l+1}s,\w v<n^{\frac{1}{2}-2\epsilon}} \w v<n^{\frac{1}{2}+\epsilon}$ if $\lambda$ has finite second moment, $\sum_{v \in \G{l+1}s,\w v<n^{\frac{1}{2}-2\epsilon}} \w v<n^{\frac{4-\beta}{2}+\epsilon}$ if $\lambda$ is power law with $2<\beta<3$.
\begin{proof}[Proof of claim, CM]
As in the proof of \Cref{lem:touchbid}, we can safely assume that we know the structure $S$ of $\N ls$. Let us sort the stubs in $\R ls$, not paired by $S$, obtaining $a_1,\dots,a_k$, and let $\boldsymbol{a}_i$ be the stub paired with $a_i$. Let $\res(a)$ be the number of stubs of the node $a$, minus $a$, and let $\X_i=\res(\boldsymbol{a}_i)$ if $\res(\boldsymbol{a}_i)<n^{\frac{1}{2}-2\epsilon}$, $0$ otherwise: clearly, $\sum_{v \in \G{l+1}s,\w v \leq n^{\frac{1}{2}-2\epsilon}} \w v \leq \sum_{i=1}^{k} \X_i$ (with equality if there are no horizontal or diagonal edges in the BFS tree). After the first $i-1$ stubs are paired, since $i<n^{\frac{1}{2}+\epsilon}$ and since the number of stubs paired in $S$ is $\O\left(n^{\frac{1}{2}+\epsilon}\log n\right)$, for each $k<n^{\frac{1}{2}-2\epsilon}$,
\begin{align*}
    \P\left(\X_i=k\right)&=\P\left(\res\left(\boldsymbol{a}_i\right)=k\right) \\
    &=\frac{\left|\left\{a \in A: a \text{ unpaired after $i$ rounds}, \res(a)=k\right\}\right|}{\left|\left\{a \in A: a \text{ unpaired after $i$ rounds}\right\}\right|} \\
    &=\frac{\left|\left\{a \in A: \res(a)=k\right\}\right|+\O\left(n^{\frac{1}{2}+\epsilon}\right)}{\left|A\right|+\O\left(n^{\frac{1}{2}+\epsilon}\right)} \\
    &=\frac{(k+1)\lambda(k+1)}{M_1(\lambda)}+\O\left(n^{-\frac{1}{2}+\epsilon}\right).
\end{align*}

Consequently, conditioned on all pairings of $a_j$ for $j<i$, $\E\left[\X_i\right]=\sum_{k=0}^{n^{\frac{1}{2}-2\epsilon}} k\frac{(k+1)\lambda(k+1)}{M_1(\lambda)}+\O(n^{-\frac{1}{2}+\epsilon}\log n)=\alpha(n)$, where $\alpha(n)=\O(1)$ if $\lambda$ has finite second moment, and $\alpha(n)=\O(n^{\frac{3-\beta}{2}})$ if $\lambda$ is power law with $2<\beta<3$. Hence, for each $\epsilon$, $\sum_{i=1}^k \X_i-i(M_1(\lambda)+\epsilon)$ is a supermartingale, and by Azuma's inequality 
\[
\P\left(\sum_{i=1}^k \X_i-k\alpha(n)\geq \alpha(n)\right) \leq \exp\left({-\frac{\alpha(n)^2}{2\sum_{i=1}^k n^{\frac{1}{2}-2\epsilon}}}\right) \leq \exp(-n^\epsilon).
\]
Then, \whp, $\sum_{i=1}^k \X_i \leq n^{\frac{1}{2}+\epsilon}(\alpha(n)+2)$, concluding the proof of the claim.
\end{proof}
\begin{proof}[Proof of claim, IRG]
The number of nodes $w$ with weight at most $n^{\frac{1}{2}-2\epsilon}$ that belong to $\G{l+1}s$ is at most $\sum_{v \in \G ls,\w v<n^{\frac{1}{2}-2\epsilon}}\sum_{w \in V} \w w\X_{v,w}$, where $\X_{v,w}=1$ with probability $f\left(\frac{\w v \w w}{M}\right)=\O\left(\frac{\w v \w w}{M}\right)$ because $\w v\w w < n^{1-\epsilon}$. Moreover, 
\[
\E\left[\sum_{v \in \G ls,\w v<n^{\frac{1}{2}-2\epsilon}}\sum_{w \in V} \w w\X_{v,w}\right] = \O\left(\r ls \frac{\sum_{v \in V} \w v^2}{n}\right)=\r ls\alpha(n)
\]

where $\alpha(n)=\O(1)$ if $\lambda$ has finite second moment, and $\alpha(n)=\O\left(n^{\frac{3-\beta}{2}}\right)$ if $\lambda$ is power law with $2<\beta<3$.

By Hoeffding inequality, 
\[
    \P\left(\sum_{v \in \G ls,\w v<n^{\frac{1}{2}-2\epsilon}}\sum_{w \in V} \w w\X_{v,w} -\r ls \alpha(n)\geq \r ls \alpha(n)\right) \leq n^{\frac{\r ls \alpha(n)}{\r ls n^{\frac{1}{2}-2\epsilon}}} \leq n^{-\epsilon}.
\]
This concludes the proof.
\end{proof}

This claim lets us conclude the proof of the lemma.
\end{proof}

\begin{proof}[Proof of \Cref{thm:bidirectionalshort}]
Let $D_s^i=\sum_{v \in \G is} \deg(v)$, $D_t^j=\sum_{w \in \G jt} \deg(w)$, and let us suppose that we have visited until level $l_s$ from $s$, until level $l_t$ from $t$, and that $D_s^{l_s},D_t^{l_t}>n^{\frac{1}{2}+2\epsilon}$. If this situation never occurs, by \Cref{thm:diameter}, the total number of visited edges is at most $\O(\log n)n^{\frac{1}{2}+2\epsilon} = \O(n^{\frac{1}{2}+3\epsilon})$, and the conclusion follows. Otherwise, again by \Cref{thm:diameter}, the number of edges visited in the two BFS trees before levels $l_s$ and $l_t$ is $\O(n^{\frac{1}{2}+3\epsilon})$. Furthermore, by \Cref{lem:degsweight}, $\r{l_s}s,\r{l_t}t>n^{\frac{1}{2}+2\epsilon}$. We claim that, without loss of generality, we can assume $\r{l_s-1}s<\epsilon\r{l_s}s$, to apply \Cref{lem:touchbid}. Indeed, if $\r{l_s-1}s$ is too big, we iteratively decrease $l_s$ until we find a neighbor verifying $\r{l_s}s>(1-\epsilon')\r{l_s-1}s$. This process can last at most $\O(\log n)$ steps, and hence it is stopped at a point $l_s$ such that $\r{l_s}s>n^{\frac{1}{2}+2\epsilon}(1-\epsilon')^{\O(\log n)} \geq n^{\frac{1}{2}+\epsilon'}$ if $\epsilon'$ is small enough. Similarly, we can suppose without loss of generality that $\r{l_t}t>(1-\epsilon')\r{l_t-1}t$. By \Cref{lem:touchbid}, $d(s,t) \leq l_s+l_t+2$, and the number of nodes needed to conclude the BFS is at most $D_s^{l_s}+D_t^{l_t}$ (note that, if we extend twice the visit from $s$, it means that $D_s^{l_s+1}<D_t^{l_t}$). By \Cref{lem:degsweight}, $D_s^{l_s} \leq n^{\epsilon}\r {l_s}s$, and by \Cref{lem:lastlevel} this value is at most $n^{\frac{1}{2}+3\epsilon}$ if $\lambda$ has finite second moment, and $n^{\frac{4-\beta}{2}+3\epsilon}$ if $\lambda$ is power law with $2<\beta<3$. We conclude that the total number of visited nodes is at most $n^{\frac{1}{2}+3\epsilon}+D_s^{l_s}+D_t^{l_t} \leq n^{\frac{1}{2}+3\epsilon}+\r{l_s}s+\r{l_t}t \leq n^{\frac{1}{2}+4\epsilon}$ (resp., $n^{\frac{4-\beta}{2}+4\epsilon}$) if $\lambda$ has finite second moment (resp., if $\lambda$ is power law with $2<\beta<3$). The theorem follows by changing the value of $\epsilon$.
\end{proof}

\section{Detailed Experimental Results} \label{app:detailedresults}
\begin{table}[H]
\caption{Detailed experimental results (undirected graphs). Empty values correspond to graphs on which the algorithm needed more than $1$ hour.}
\resizebox{\linewidth}{!}{
\begin{tabular}{|l|r|r|r|r|r|r|r|r|r|}
\hline
 & \multicolumn{4}{c|}{Number of iterations} & \multicolumn{4}{c|}{Time (s)} & Edges \\
Graph & \cad & \rk & \abraaut & \abrag  & \cad & \rk & \abraaut & \abrag & \cad \\ 
\hline
\multicolumn{10}{|l|}{$\lambda = 0.005$} \\
\hline
advogato & 64427 & 126052 & 174728 & 185998 & 0.193 & 11.450 & 9.557 & 10.498 & 261.2 \\ 
as20000102 & 115797 & 126052 & 18329844 & 4126626 & 0.231 & 6.990 & 611.584 & 136.764 & 377.6 \\ 
ca-GrQc & 61611 & 146052 & 142982 & 129165 & 0.126 & 5.574 & 3.500 & 2.839 & 353.4 \\ 
ca-HepTh & 31735 & 146052 & 121587 & 129165 & 0.222 & 14.921 & 7.389 & 8.168 & 9.9 \\ 
C\_elegans & 69729 & 146052 & 204634 & 185998 & 0.132 & 6.876 & 5.693 & 5.261 & 270.7 \\ 
com-amazon.all & 40711 & 166052 & 69708 & 74747 & 0.340 & 122.020 & 12.011 & 11.849 & 21.9 \\ 
dip20090126\_MAX & 156552 & 166052 &  &  & 1.374 & 34.595 &  &  & 15354.9 \\ 
D\_melanogaster & 51227 & 126052 & 144680 & 154998 & 0.123 & 19.253 & 15.061 & 16.882 & 520.8 \\ 
email-Enron & 74745 & 146052 & 257989 & 267838 & 0.280 & 79.296 & 101.529 & 106.278 & 1408.0 \\ 
HC-BIOGRID & 78804 & 146052 & 245780 & 223198 & 0.177 & 7.751 & 7.534 & 6.951 & 713.2 \\ 
Homo\_sapiens & 60060 & 146052 & 156973 & 154998 & 0.151 & 32.078 & 23.716 & 24.449 & 643.8 \\ 
hprd\_pp & 59125 & 146052 & 151499 & 154998 & 0.127 & 18.323 & 13.425 & 13.458 & 456.4 \\ 
Mus\_musculus & 92081 & 146052 & 504669 & 385688 & 0.168 & 4.058 & 7.723 & 6.083 & 226.6 \\ 
oregon1\_010526 & 114829 & 126052 & 6798931 & 2865712 & 0.228 & 13.281 & 442.370 & 185.711 & 681.6 \\ 
oregon2\_010526 & 115764 & 126052 & 5714183 & 2865712 & 0.236 & 15.823 & 452.554 & 229.234 & 822.2 \\ 
\hline
\multicolumn{10}{|l|}{$\lambda = 0.010$} \\
\hline
advogato & 19811 & 31513 & 47076 & 48243 & 0.081 & 2.804 & 2.576 & 2.788 & 258.2 \\ 
as20000102 & 29062 & 31513 & 2688614 & 1070372 & 0.071 & 1.777 & 88.886 & 35.049 & 377.3 \\ 
ca-GrQc & 18535 & 36513 & 37529 & 33501 & 0.049 & 1.417 & 0.987 & 0.753 & 350.6 \\ 
ca-HepTh & 13761 & 36513 & 31721 & 33501 & 0.188 & 3.771 & 2.078 & 2.275 & 10.0 \\ 
C\_elegans & 19888 & 36513 & 54327 & 48243 & 0.048 & 1.803 & 1.586 & 1.483 & 269.4 \\ 
com-amazon.all & 14641 & 41513 & 18007 & 19386 & 0.312 & 31.004 & 5.196 & 7.623 & 21.5 \\ 
dip20090126\_MAX & 39314 & 41513 &  &  & 0.395 & 8.578 &  &  & 15301.7 \\ 
D\_melanogaster & 15136 & 31513 & 37219 & 40202 & 0.063 & 4.983 & 3.891 & 4.715 & 519.9 \\ 
email-Enron & 21637 & 36513 & 65392 & 69471 & 0.198 & 19.877 & 24.997 & 27.296 & 1387.2 \\ 
HC-BIOGRID & 22924 & 36513 & 62413 & 57892 & 0.052 & 1.979 & 1.989 & 1.906 & 712.5 \\ 
Homo\_sapiens & 20273 & 36513 & 41006 & 40202 & 0.085 & 7.876 & 6.442 & 6.636 & 642.7 \\ 
hprd\_pp & 18403 & 36513 & 39994 & 40202 & 0.074 & 4.348 & 4.097 & 3.714 & 456.4 \\ 
Mus\_musculus & 25146 & 36513 & 130384 & 100040 & 0.061 & 1.055 & 1.965 & 1.718 & 223.9 \\ 
oregon1\_010526 & 30514 & 31513 & 1104167 & 743313 & 0.087 & 3.254 & 70.383 & 47.740 & 683.3 \\ 
oregon2\_010526 & 29117 & 31513 & 954515 & 743313 & 0.088 & 3.983 & 73.942 & 59.103 & 822.1 \\ 
\hline
\multicolumn{10}{|l|}{$\lambda = 0.015$} \\
\hline
advogato & 9570 & 14006 & 21027 & 22204 & 0.050 & 1.428 & 1.227 & 1.299 & 261.0 \\ 
as20000102 & 13035 & 14006 & 705483 & 492651 & 0.047 & 0.776 & 22.939 & 16.136 & 377.6 \\ 
ca-GrQc & 8668 & 16228 & 17419 & 15419 & 0.031 & 0.637 & 0.493 & 0.361 & 345.8 \\ 
ca-HepTh & 7524 & 16228 & 15002 & 15419 & 0.167 & 1.641 & 0.939 & 1.050 & 11.5 \\ 
C\_elegans & 10956 & 16228 & 25233 & 22204 & 0.034 & 0.782 & 0.740 & 0.732 & 267.6 \\ 
com-amazon.all & 8228 & 18451 &  & 15419 & 0.301 & 13.814 &  & 7.785 & 21.9 \\ 
dip20090126\_MAX & 17578 & 18451 &  &  & 0.203 & 3.851 &  &  & 15197.2 \\ 
D\_melanogaster & 9350 & 14006 & 17229 & 18503 & 0.053 & 2.216 & 1.904 & 2.182 & 519.3 \\ 
email-Enron & 11209 & 16228 & 29134 & 31974 & 0.170 & 8.845 & 10.510 & 12.423 & 1367.4 \\ 
HC-BIOGRID & 12694 & 16228 & 28805 & 26645 & 0.043 & 0.858 & 0.946 & 0.947 & 708.6 \\ 
Homo\_sapiens & 10142 & 16228 & 18491 & 18503 & 0.072 & 3.717 & 3.076 & 3.061 & 640.4 \\ 
hprd\_pp & 10659 & 16228 & 17969 & 18503 & 0.056 & 1.919 & 1.719 & 1.752 & 451.5 \\ 
Mus\_musculus & 11825 & 16228 & 59756 & 46043 & 0.033 & 0.458 & 0.906 & 0.812 & 222.8 \\ 
oregon1\_010526 & 13662 & 14006 & 426845 & 342118 & 0.056 & 1.522 & 26.420 & 21.871 & 681.4 \\ 
oregon2\_010526 & 13024 & 14006 & 333638 & 342118 & 0.060 & 1.773 & 26.070 & 27.298 & 833.6 \\ 
\hline
    \end{tabular}
}
\end{table} 

\begin{table}[H]
\resizebox{\linewidth}{!}{
\begin{tabular}{|l|r|r|r|r|r|r|r|r|r|}
\hline
 & \multicolumn{4}{c|}{Number of iterations} & \multicolumn{4}{c|}{Time (s)} & Edges \\
Graph & \cad & \rk & \abraaut & \abrag  & \cad & \rk & \abraaut & \abrag & \cad \\ 
\hline
\multicolumn{10}{|l|}{$\lambda = 0.020$} \\
\hline
advogato & 5874 & 7879 & 11993 & 12915 & 0.054 & 0.710 & 0.665 & 0.765 & 260.3 \\ 
as20000102 & 7436 & 7879 & 312581 & 238814 & 0.037 & 0.441 & 10.066 & 7.819 & 376.2 \\ 
ca-GrQc & 5313 & 9129 & 9939 & 10762 & 0.032 & 0.356 & 0.293 & 0.268 & 347.9 \\ 
ca-HepTh & 5115 & 9129 & 8708 & 8968 & 0.191 & 0.891 & 0.694 & 0.611 & 10.5 \\ 
C\_elegans & 7172 & 9129 & 14871 & 12915 & 0.030 & 0.439 & 0.436 & 0.439 & 263.5 \\ 
com-amazon.all & 5467 & 10379 & 12232 & 10762 & 0.331 & 7.683 & 4.338 & 5.459 & 17.9 \\ 
dip20090126\_MAX & 9966 & 10379 &  &  & 0.148 & 2.165 &  &  & 15188.3 \\ 
D\_melanogaster & 5610 & 7879 & 10201 & 10762 & 0.056 & 1.236 & 1.265 & 1.306 & 520.9 \\ 
email-Enron & 7458 & 9129 & 16443 & 15498 & 0.174 & 4.916 & 6.102 & 6.034 & 1371.7 \\ 
HC-BIOGRID & 8459 & 9129 & 17406 & 15498 & 0.026 & 0.505 & 0.602 & 0.582 & 716.6 \\ 
Homo\_sapiens & 6292 & 9129 & 10481 & 10762 & 0.064 & 1.944 & 1.672 & 1.814 & 644.8 \\ 
hprd\_pp & 6611 & 9129 & 10501 & 10762 & 0.050 & 1.089 & 0.930 & 1.050 & 449.8 \\ 
Mus\_musculus & 7227 & 9129 & 31634 & 26782 & 0.026 & 0.255 & 0.507 & 0.532 & 221.0 \\ 
oregon1\_010526 & 7733 & 7879 & 220948 & 199011 & 0.051 & 0.863 & 13.584 & 12.989 & 679.2 \\ 
oregon2\_010526 & 7381 & 7879 & 152242 & 165842 & 0.059 & 1.031 & 11.676 & 13.290 & 836.0 \\ 
\hline
\multicolumn{10}{|l|}{$\lambda = 0.025$} \\
\hline
advogato & 3883 & 5043 & 7439 & 7110 & 0.052 & 0.450 & 0.421 & 0.468 & 263.4 \\ 
as20000102 & 4829 & 5043 & 130506 & 157779 & 0.033 & 0.285 & 4.097 & 5.108 & 373.5 \\ 
ca-GrQc & 3982 & 5843 & 6427 & 5925 & 0.028 & 0.242 & 0.180 & 0.162 & 342.1 \\ 
ca-HepTh & 3773 & 5843 & 6016 & 5925 & 0.176 & 0.573 & 0.374 & 0.416 & 11.8 \\ 
C\_elegans & 4477 & 5843 & 9557 & 8532 & 0.025 & 0.292 & 0.293 & 0.293 & 266.6 \\ 
com-amazon.all & 4059 & 6643 & 58995 & 14745 & 0.338 & 4.744 & 9.644 & 7.217 & 21.3 \\ 
dip20090126\_MAX & 6457 & 6643 &  &  & 0.125 & 1.397 &  &  & 15193.8 \\ 
D\_melanogaster & 3993 & 5043 & 6279 & 7110 & 0.056 & 0.793 & 0.827 & 0.870 & 522.6 \\ 
email-Enron & 4576 & 5843 & 11001 & 12287 & 0.574 & 3.289 & 3.888 & 4.705 & 1381.5 \\ 
HC-BIOGRID & 5940 & 5843 & 11109 & 10239 & 0.029 & 0.321 & 0.414 & 0.404 & 714.0 \\ 
Homo\_sapiens & 4796 & 5843 & 7109 & 7110 & 0.077 & 1.245 & 1.154 & 1.215 & 647.2 \\ 
hprd\_pp & 5071 & 5843 & 6772 & 7110 & 0.052 & 0.687 & 0.579 & 0.647 & 446.3 \\ 
Mus\_musculus & 4477 & 5843 & 18626 & 17694 & 0.026 & 0.168 & 0.302 & 0.385 & 219.8 \\ 
oregon1\_010526 & 5027 & 5043 & 92520 & 109568 & 0.058 & 0.516 & 5.762 & 7.014 & 681.0 \\ 
oregon2\_010526 & 4763 & 5043 & 86287 & 91306 & 0.050 & 0.638 & 7.140 & 7.420 & 847.5 \\ 
\hline
\multicolumn{10}{|l|}{$\lambda = 0.030$} \\
\hline
advogato & 3256 & 3502 & 5521 & 5090 & 0.048 & 0.361 & 0.335 & 0.322 & 260.6 \\ 
as20000102 & 3388 & 3502 & 122988 & 94140 & 0.029 & 0.199 & 3.899 & 3.182 & 378.7 \\ 
ca-GrQc & 2981 & 4057 & 4686 & 4241 & 0.025 & 0.169 & 0.145 & 0.175 & 344.7 \\ 
ca-HepTh & 2992 & 4057 & 4022 & 4241 & 0.190 & 0.435 & 0.286 & 0.341 & 7.9 \\ 
C\_elegans & 3707 & 4057 & 6905 & 6108 & 0.026 & 0.198 & 0.218 & 0.217 & 265.9 \\ 
com-amazon.all & 3157 & 4613 & 39917 & 12668 & 0.330 & 3.631 & 8.491 & 6.852 & 17.5 \\ 
dip20090126\_MAX & 4499 & 4613 & 12373086 &  & 0.300 & 0.972 & 1958.083 &  & 15199.0 \\ 
D\_melanogaster & 2893 & 3502 & 4883 & 5090 & 0.052 & 0.562 & 0.620 & 0.807 & 510.4 \\ 
email-Enron & 3619 & 4057 & 7321 & 7330 & 0.172 & 2.735 & 2.724 & 2.806 & 1399.7 \\ 
HC-BIOGRID & 3883 & 4057 & 7499 & 7330 & 0.024 & 0.367 & 0.316 & 0.307 & 720.8 \\ 
Homo\_sapiens & 3322 & 4057 & 4982 & 5090 & 0.066 & 0.897 & 0.842 & 0.877 & 654.2 \\ 
hprd\_pp & 3355 & 4057 & 5028 & 5090 & 0.048 & 0.478 & 0.458 & 0.503 & 448.8 \\ 
Mus\_musculus & 3806 & 4057 & 14290 & 10556 & 0.033 & 0.127 & 0.237 & 0.233 & 221.4 \\ 
oregon1\_010526 & 3542 & 3502 & 85854 & 78450 & 0.052 & 0.366 & 5.402 & 5.039 & 675.7 \\ 
oregon2\_010526 & 3355 & 3502 & 61841 & 65375 & 0.048 & 0.509 & 4.972 & 5.302 & 822.8 \\
\hline
    \end{tabular}
}
\end{table}

\begin{table}
\caption{Detailed experimental results (directed graphs). Empty values correspond to graphs on which the algorithm needed more than $1$ hour.}
\resizebox{\linewidth}{!}{
\begin{tabular}{|l|r|r|r|r|r|r|r|r|r|}
\hline
 & \multicolumn{4}{c|}{Number of iterations} & \multicolumn{4}{c|}{Time (s)} & Edges \\
Graph & \cad & \rk & \abraaut & \abrag  & \cad & \rk & \abraaut & \abrag & \cad \\ 
\hline
\multicolumn{10}{|l|}{$\lambda = 0.005$} \\
\hline
as-caida20071105 & 103488 & 146052 & 546951 & 462826 & 0.253 & 35.652 & 96.312 & 85.201 & 1066.4 \\ 
cfinder-google & 137313 & 146052 &  &  & 0.820 & 14.190 &  &  & 554.4 \\ 
cit-HepTh & 98054 & 166052 & 481476 & 462826 & 0.579 & 22.651 & 38.339 & 37.720 & 5773.1 \\ 
ego-gplus & 37862 & 66052 &  & 2388093 & 0.136 & 6.266 &  & 11.912 & 1.9 \\ 
ego-twitter & 37125 & 66052 &  & 154998 & 0.178 & 6.181 &  & 4.804 & 2.3 \\ 
freeassoc & 41602 & 166052 & 89424 & 89697 & 0.116 & 9.384 & 1.036 & 0.997 & 223.5 \\ 
lasagne-spanishbook & 112266 & 146052 & 8918751 & 4126626 & 0.250 & 17.374 & 687.815 & 318.784 & 552.8 \\ 
opsahl-openflights & 73744 & 146052 & 200164 & 185998 & 0.179 & 6.191 & 5.165 & 4.849 & 431.1 \\ 
p2p-Gnutella31 & 39193 & 166052 & 81335 & 89697 & 0.254 & 50.542 & 10.213 & 10.662 & 162.1 \\ 
polblogs & 71423 & 126052 & 387278 & 321406 & 0.174 & 1.165 & 3.522 & 3.017 & 190.3 \\ 
soc-Epinions1 & 58223 & 146052 & 109607 & 107637 & 0.671 & 100.516 & 62.524 & 62.167 & 671.9 \\ 
subelj-cora-cora & 68112 & 186052 & 180740 & 185998 & 0.185 & 19.012 & 8.464 & 8.873 & 440.4 \\ 
subelj-jdk-jdk & 42361 & 146052 & 84549 & 89697 & 0.110 & 2.955 & 0.230 & 0.257 & 51.5 \\ 
subelj-jung-j-jung-j & 43637 & 126052 & 84225 & 89697 & 0.216 & 2.397 & 0.238 & 0.211 & 45.9 \\ 
wiki-Vote & 47003 & 126052 & 100153 & 107637 & 0.131 & 5.916 & 2.990 & 3.219 & 162.4 \\ 
\hline
\multicolumn{10}{|l|}{$\lambda = 0.010$} \\
\hline
as-caida20071105 & 30382 & 36513 & 132997 & 120048 & 0.135 & 8.902 & 22.251 & 20.315 & 1066.1 \\ 
cfinder-google & 34452 & 36513 &  &  & 0.156 & 3.664 &  &  & 553.2 \\ 
cit-HepTh & 27203 & 41513 & 117633 & 120048 & 0.255 & 5.654 & 8.803 & 9.677 & 5798.8 \\ 
ego-gplus & 13123 & 16513 &  & 4602412 & 0.085 & 1.584 &  & 22.510 & 2.3 \\ 
ego-twitter & 13310 & 16513 &  & 83366 & 0.086 & 1.518 &  & 3.500 & 2.2 \\ 
freeassoc & 13222 & 41513 & 23586 & 23264 & 0.080 & 2.335 & 0.238 & 0.227 & 220.7 \\ 
lasagne-spanishbook & 32527 & 36513 & 1366576 & 1070372 & 0.101 & 4.339 & 104.916 & 83.610 & 553.4 \\ 
opsahl-openflights & 22473 & 36513 & 52196 & 48243 & 0.059 & 1.475 & 1.348 & 1.339 & 432.0 \\ 
p2p-Gnutella31 & 13101 & 41513 & 21567 & 23264 & 0.192 & 12.950 & 2.677 & 2.831 & 162.1 \\ 
polblogs & 22286 & 31513 & 101466 & 83366 & 0.046 & 0.298 & 1.078 & 0.834 & 190.6 \\ 
soc-Epinions1 & 17061 & 36513 & 28493 & 27917 & 0.320 & 27.194 & 16.516 & 15.974 & 659.5 \\ 
subelj-cora-cora & 23078 & 46513 & 47936 & 48243 & 0.128 & 4.797 & 1.988 & 2.101 & 432.4 \\ 
subelj-jdk-jdk & 14047 & 36513 & 22038 & 23264 & 0.066 & 0.734 & 0.099 & 0.075 & 52.2 \\ 
subelj-jung-j-jung-j & 14894 & 36513 & 22266 & 23264 & 0.064 & 0.696 & 0.113 & 0.083 & 46.4 \\ 
wiki-Vote & 17380 & 31513 & 26352 & 27917 & 0.088 & 1.446 & 0.792 & 0.870 & 155.7 \\ 
\hline
\multicolumn{10}{|l|}{$\lambda = 0.015$} \\
\hline
as-caida20071105 & 14157 & 16228 & 55049 & 55252 & 0.477 & 3.963 & 8.518 & 8.914 & 1059.6 \\ 
cfinder-google & 15400 & 16228 &  &  & 0.123 & 1.666 &  &  & 558.1 \\ 
cit-HepTh & 13002 & 18451 & 47035 & 46043 & 0.232 & 2.529 & 3.807 & 3.766 & 5883.0 \\ 
ego-gplus & 7205 & 7340 &  & 2118317 & 0.080 & 0.710 &  & 12.808 & 2.2 \\ 
ego-twitter & 7403 & 7340 & 1958981 & 114573 & 0.082 & 0.704 & 14.021 & 5.304 & 2.3 \\ 
freeassoc & 7095 & 18451 & 10956 & 10707 & 0.297 & 1.072 & 0.115 & 0.110 & 222.0 \\ 
lasagne-spanishbook & 14542 & 16228 & 437041 & 410542 & 0.068 & 1.936 & 34.098 & 33.153 & 552.8 \\ 
opsahl-openflights & 11550 & 16228 & 24433 & 22204 & 0.034 & 0.649 & 0.643 & 0.648 & 433.9 \\ 
p2p-Gnutella31 & 7227 & 18451 & 10002 & 10707 & 0.190 & 5.732 & 1.317 & 1.444 & 157.1 \\ 
polblogs & 10296 & 14006 & 46648 & 38369 & 0.029 & 0.136 & 0.516 & 0.435 & 189.5 \\ 
soc-Epinions1 & 9273 & 16228 & 13571 & 12849 & 0.450 & 12.115 & 7.661 & 7.629 & 662.0 \\ 
subelj-cora-cora & 11297 & 20673 & 20940 & 22204 & 0.502 & 2.135 & 0.937 & 1.073 & 445.6 \\ 
subelj-jdk-jdk & 8360 & 14006 & 10045 & 10707 & 0.052 & 0.288 & 0.080 & 0.049 & 51.6 \\ 
subelj-jung-j-jung-j & 8712 & 16228 & 10319 & 10707 & 0.046 & 0.312 & 0.068 & 0.042 & 45.6 \\ 
wiki-Vote & 8668 & 14006 & 12406 & 12849 & 0.408 & 0.659 & 0.380 & 0.429 & 152.6 \\ 
\hline
\end{tabular}}
\end{table}

\begin{table}[H]
\resizebox{\linewidth}{!}{
\begin{tabular}{|l|r|r|r|r|r|r|r|r|r|}
\hline
 & \multicolumn{4}{c|}{Number of iterations} & \multicolumn{4}{c|}{Time (s)} & Edges \\
Graph & \cad & \rk & \abraaut & \abrag  & \cad & \rk & \abraaut & \abrag & \cad \\ 
\hline
\multicolumn{10}{|l|}{$\lambda = 0.020$} \\
\hline
as-caida20071105 & 9086 & 9129 & 31242 & 32139 & 0.104 & 2.226 & 4.954 & 5.087 & 1064.2 \\ 
cfinder-google & 8745 & 9129 &  &  & 0.353 & 0.946 &  &  & 551.9 \\ 
cit-HepTh & 8679 & 10379 & 27755 & 32139 & 1.249 & 1.442 & 2.225 & 2.684 & 5758.0 \\ 
ego-gplus & 4785 & 4129 &  & 1478684 & 0.081 & 0.395 &  & 9.234 & 2.6 \\ 
ego-twitter & 4950 & 7879 &  & 138201 & 0.083 & 0.743 &  & 5.079 & 2.4 \\ 
freeassoc & 4268 & 10379 & 6509 & 6227 & 0.065 & 0.609 & 0.078 & 0.073 & 216.4 \\ 
lasagne-spanishbook & 8338 & 9129 & 294793 & 286577 & 0.058 & 1.074 & 22.405 & 22.468 & 555.0 \\ 
opsahl-openflights & 7392 & 9129 & 14202 & 12915 & 0.029 & 0.364 & 0.390 & 0.391 & 432.3 \\ 
p2p-Gnutella31 & 4697 & 10379 & 5700 & 6227 & 0.190 & 3.162 & 0.695 & 0.816 & 156.7 \\ 
polblogs & 6325 & 7879 & 25593 & 22318 & 0.023 & 0.076 & 0.283 & 0.252 & 188.4 \\ 
soc-Epinions1 & 5489 & 9129 & 7686 & 7473 & 0.457 & 6.738 & 4.506 & 4.335 & 651.8 \\ 
subelj-cora-cora & 6325 & 11629 & 12437 & 12915 & 0.500 & 1.203 & 0.571 & 0.520 & 450.8 \\ 
subelj-jdk-jdk & 5456 & 9129 & 6070 & 6227 & 0.191 & 0.192 & 0.062 & 0.044 & 52.3 \\ 
subelj-jung-j-jung-j & 5643 & 9129 &  & 6227 & 0.217 & 0.176 &  & 0.045 & 46.6 \\ 
wiki-Vote & 4939 & 7879 & 7125 & 7473 & 0.075 & 0.368 & 0.221 & 0.259 & 152.2 \\ 
\hline
\multicolumn{10}{|l|}{$\lambda = 0.025$} \\
\hline
as-caida20071105 & 5723 & 5843 & 21020 & 21233 & 0.022 & 1.465 & 3.129 & 3.340 & 1093.4 \\ 
cfinder-google & 6275 & 5843 &  &  & 0.019 & 0.648 &  &  & 758.0 \\ 
cit-HepTh & 5206 & 6643 & 15915 & 21233 & 0.034 & 0.940 & 1.351 & 1.891 & 6130.5 \\ 
ego-gplus & 2989 & 5043 &  & 4200646 & 0.013 & 0.485 &  & 20.309 & 2.6 \\ 
ego-twitter & 2958 & 2643 &  & 157779 & 0.012 & 0.248 &  & 6.291 & 2.4 \\ 
freeassoc & 2804 & 6643 & 4285 & 4114 & 0.009 & 0.399 & 0.061 & 0.058 & 261.5 \\ 
lasagne-spanishbook & 5409 & 5043 & 129999 & 131482 & 0.013 & 0.592 & 10.040 & 10.221 & 626.1 \\ 
opsahl-openflights & 4557 & 5843 & 10116 & 8532 & 0.009 & 0.236 & 0.290 & 0.267 & 561.3 \\ 
p2p-Gnutella31 & 3069 & 6643 & 3931 & 4114 & 0.043 & 2.149 & 0.590 & 0.663 & 176.8 \\ 
polblogs & 3880 & 5043 & 15986 & 14745 & 0.007 & 0.049 & 0.185 & 0.176 & 241.9 \\ 
soc-Epinions1 & 3689 & 5843 & 5060 & 4937 & 0.188 & 4.158 & 2.798 & 2.791 & 888.1 \\ 
subelj-cora-cora & 5264 & 7443 & 7699 & 8532 & 0.020 & 0.781 & 0.360 & 0.408 & 436.5 \\ 
subelj-jdk-jdk & 3201 & 5843 & 9428 & 4937 & 0.008 & 0.122 & 0.065 & 0.036 & 57.2 \\ 
subelj-jung-j-jung-j & 3168 & 5043 & 13471 & 5925 & 0.007 & 0.098 & 0.057 & 0.045 & 57.7 \\ 
wiki-Vote & 3265 & 5043 & 4566 & 4937 & 0.009 & 0.241 & 0.137 & 0.178 & 174.7 \\ 
\hline
\multicolumn{10}{|l|}{$\lambda = 0.030$} \\
\hline
as-caida20071105 & 3956 & 4057 & 12696 & 15202 & 0.017 & 1.029 & 1.973 & 2.434 & 1285.2 \\ 
cfinder-google & 4419 & 4057 &  &  & 0.013 & 0.412 &  &  & 770.5 \\ 
cit-HepTh & 4062 & 4613 & 13172 & 12668 & 0.033 & 0.672 & 1.195 & 1.059 & 6131.6 \\ 
ego-gplus & 2434 & 1835 &  & 4330990 & 0.009 & 0.188 &  & 21.395 & 3.1 \\ 
ego-twitter & 2270 & 1835 & 98839 & 135562 & 0.008 & 0.174 & 4.909 & 5.510 & 2.2 \\ 
freeassoc & 2105 & 4613 & 3008 & 3534 & 0.006 & 0.285 & 0.101 & 0.091 & 250.7 \\ 
lasagne-spanishbook & 3820 & 4057 & 158028 & 94140 & 0.010 & 0.487 & 12.564 & 7.855 & 656.8 \\ 
opsahl-openflights & 3450 & 4057 & 6556 & 6108 & 0.007 & 0.165 & 0.184 & 0.195 & 481.4 \\ 
p2p-Gnutella31 & 2367 & 4613 & 2874 & 2945 & 0.036 & 1.412 & 0.422 & 0.445 & 166.5 \\ 
polblogs & 3567 & 3502 & 11357 & 8796 & 0.007 & 0.036 & 0.151 & 0.122 & 207.9 \\ 
soc-Epinions1 & 2659 & 4057 & 3585 & 3534 & 0.312 & 3.211 & 2.186 & 2.046 & 918.3 \\ 
subelj-cora-cora & 3790 & 5169 & 5681 & 5090 & 0.016 & 0.564 & 0.272 & 0.265 & 422.6 \\ 
subelj-jdk-jdk & 2425 & 4057 & 25575 & 5090 & 0.006 & 0.097 & 0.100 & 0.064 & 57.4 \\ 
subelj-jung-j-jung-j & 2436 & 3502 & 43584 & 5090 & 0.006 & 0.079 & 0.140 & 0.059 & 57.0 \\ 
wiki-Vote & 2633 & 3502 & 3467 & 3534 & 0.006 & 0.188 & 0.148 & 0.149 & 188.2 \\
\hline
\end{tabular}}
\end{table}

\section{Wikipedia and IMDB Results} \label{sec:wikipediaimdb}
\nopagebreak
In this section, we report our results on the Wikipedia citation network, and on all snapshots of the IMDB actors collaboration network. In the ranking column, we report one number if the position in the ranking is guaranteed with probability $0.9$, otherwise we report a lower and an upper bound, which hold with the same probability.

We remark that, as for the IMDB database, the top-$k$ betweenness centralities of a single snapshot of a similar graph (\texttt{hollywood-2009} in \cite{sebagraph}) have been previously computed exactly, with one week of computation on a $40$-core machine \cite{seba}.

\subsection{The Results on the IMDB Graph}
\label{sec:imdb_results}

In 2014, the most central actor is Ron Jeremy, who is listed in the Guinness Book of World Records for ``Most Appearances in Adult Films'', with more than 2000 appearances. Among his non-adult ones, we mention  \textit{The Godfather Part III}, \textit{Ghostbusters}, \textit{Crank: High Voltage} and \textit{Family Guy}\footnote{The latter is a TV-series, which are not taken into account in our data.}. 
His topmost centrality in the actor collaboration network has been previously observed by similar experiments on betweenness centrality \cite{seba}.    
Indeed, around 3 actors out of 100 in the IMDB database played in adult movies, which explains why the high number of appearances of Ron Jeremy both in the adult and non-adult film industry rises his betweenness to the top.  

The second most-central actor is Lloyd Kaufman, which is best known as a co-founder of \textit{Troma Entertainment Film Studio} and as the director of many of their feature films, including the cult movie \textit{The Toxic Avenger}. His high betweenness score is likely due to his central role in the low-budget independent film industry.

The third ``actor'' is the historical German dictator Adolf Hitler, since his appearances in several historical footages, that were re-used in several movies (e.g. in \emph{The Imitation Game}), are credited by IMDB as cameo role. 
Indeed, he appears among the topmost actors since the 1984 snapshot, being the first one in the 1989 and 1994 ones, and during those years many movies about the World War II were produced. 

Observe that the betweenness centrality measure on our graph does not discriminate between important and marginal roles. For example, the actress Bess Flowers, who appears among the top actors in the snapshots from 1959 to 1979, rarely played major roles, but she appeared in over 700 movies in her 41 years career.

\subsection{The Results on the Wikipedia Graph}
\label{sec:wiki_results}

All topmost pages in the betweenness centrality ranking, except for the World War II, are countries. This is not surprising if we consider that, for most topics (such as important people or events), the corresponding Wikipedia page refers to their geographical context (since it mentions the country of origin of the given person or where a given event took place). 
It is also worth noting the correlation between the high centrality of the \emph{World War II} Wikipedia page and that of Adolf Hitler in the IMDB graph. 

Interestingly, a similar ranking is obtained by considering the closeness centrality measure in the inverse graph, where a link from page $p_1$ to page $p_2$ exists if a link to page $p_1$ appears in page $p_2$ \cite{Bergamini2016ComputingTC}.
However, in contrast with the results in \cite{Bergamini2016ComputingTC} when edges are oriented in the usual way, the pages about specific years do not appear in the top ranking.  
We note that the betweenness centrality of a node in a directed graph does not change if the orientation of all edges is flipped. 

Finally, the most important pages is the United States, confirming a common conjecture. Indeed, in \url{http://wikirank.di.unimi.it/}, it is shown that the United States are the center according to harmonic centrality, and many other measures. 
Further evidence for this conjecture comes from the Six Degree of Wikipedia game (\url{http://thewikigame.com/6-degrees-of-wikipedia}), where a player is asked to go from one page to the other following the smallest possible number of links: a hard variant of this game forces the player not to pass from the \emph{United States} page, which is considered to be central. 
Our results thus confirm that the conjecture is indeed true for the betweenness centrality measure.

\begin{table}[H]
\caption{The top-$k$ betweenness centralities of the Wikipedia graph computed by \cad with $\delta=0.1$ and $\lambda = 0.0002$.}
\label{tab:wikipedia}
\centering
\begin{tabular}{|l|l|r|r|r|}
\hline
Ranking & Wikipedia page & Lower bound & Estimated betweenness & Upper bound \\ 
\hline
1) & United States & 0.046278 & 0.047173 & 0.048084\\
2)  & France & 0.019522 & 0.020103 & 0.020701\\
3)  & United Kingdom & 0.017983 & 0.018540 & 0.019115\\
4)  & England & 0.016348 & 0.016879 & 0.017428\\
5-6)  & Poland & 0.012092 & 0.012287 & 0.012486\\
5-6)  & Germany & 0.011930 & 0.012124 & 0.012321\\
7)  & India & 0.009683 & 0.010092 & 0.010518\\
8-12) & World War II & 0.008870 & 0.009065 & 0.009265\\
8-12) & Russia & 0.008660 & 0.008854 & 0.009053\\
8-12) & Italy & 0.008650 & 0.008845 & 0.009045\\
8-12) & Canada & 0.008624 & 0.008819 & 0.009018\\
8-12) & Australia & 0.008620 & 0.008814 & 0.009013\\
\hline
\end{tabular}
\end{table}

\begin{table}[H]
\caption{The top-$k$ betweenness centralities of a snapshot of the IMDB collaboration network taken at the end of 1939 (69011 nodes), computed by \cad with $\delta=0.1$ and $\lambda = 0.0002$.}
\centering
\begin{tabular}{|l|l|r|r|r|}
\hline
Ranking & Actor & Lower bound & Estimated betweenness & Upper bound \\ 
\hline
1) & Meyer, Torben & 0.022331 & 0.022702 & 0.023049\\
2) & Roulien, Raul & 0.021361 & 0.021703 & 0.022071\\
3) & Myzet, Rudolf & 0.014229 & 0.014525 & 0.014747\\
4) & Sten, Anna & 0.013245 & 0.013460 & 0.013723\\
5) & Negri, Pola & 0.012509 & 0.012768 & 0.012943\\
6-7) & Jung, Shia & 0.012250 & 0.012379 & 0.012509\\
6-7) & Ho, Tai-Hau & 0.012195 & 0.012324 & 0.012454\\
8) & Goetzke, Bernhard & 0.010721 & 0.010978 & 0.011201\\
9-10) & Yamamoto, Togo & 0.010095 & 0.010224 & 0.010354\\
9-10) & Kamiyama, S\=ojin & 0.010087 & 0.010215 & 0.010344\\
\hline
\end{tabular}
\end{table}

\begin{table}[H]
\caption{The top-$k$ betweenness centralities of a snapshot of the IMDB collaboration network taken at the end of 1944 (83068 nodes), computed by \cad with $\delta=0.1$ and $\lambda = 0.0002$.}
\centering
\begin{tabular}{|l|l|r|r|r|}
\hline
Ranking & Actor & Lower bound & Estimated betweenness & Upper bound \\ 
\hline
1) & Meyer, Torben & 0.018320 & 0.018724 & 0.019136\\
2) & Kamiyama, S\=ojin & 0.012629 & 0.012964 & 0.013308\\
3-4) & Jung, Shia & 0.010751 & 0.010901 & 0.011053\\
3-4) & Ho, Tai-Hau & 0.010704 & 0.010854 & 0.011005\\
5) & Myzet, Rudolf & 0.010365 & 0.010514 & 0.010666\\
6-7) & Sten, Anna & 0.009778 & 0.009928 & 0.010080\\
6-7) & Goetzke, Bernhard & 0.009766 & 0.009915 & 0.010066\\
8) & Yamamoto, Togo & 0.009108 & 0.009327 & 0.009539\\
9) & Par\`is, Manuel & 0.008649 & 0.008859 & 0.009108\\
10) & Hayakawa, Sessue & 0.007916 & 0.008158 & 0.008369\\
\hline
\end{tabular}
\end{table}

\begin{table}[H]
\caption{The top-$k$ betweenness centralities of a snapshot of the IMDB collaboration network taken at the end of 1949 (97824 nodes), computed by \cad with $\delta=0.1$ and $\lambda = 0.0002$.}
\centering
\begin{tabular}{|l|l|r|r|r|}
\hline
Ranking & Actor & Lower bound & Estimated betweenness & Upper bound \\ 
\hline
1) & Meyer, Torben & 0.016139 & 0.016679 & 0.017236\\
2) & Kamiyama, S\=ojin & 0.012351 & 0.012822 & 0.013312\\
3) & Par\`is, Manuel & 0.011104 & 0.011552 & 0.011861\\
4) & Yamamoto, Togo & 0.010342 & 0.010639 & 0.011086\\
5-6) & Jung, Shia & 0.008926 & 0.009120 & 0.009318\\
5-6) & Goetzke, Bernhard & 0.008567 & 0.008762 & 0.008962\\
7-9) & Paananen, Tuulikki & 0.008147 & 0.008341 & 0.008539\\
7-9) & Sten, Anna & 0.007969 & 0.008164 & 0.008363\\
7-9) & Mayer, Ruby & 0.007967 & 0.008162 & 0.008362\\
10-12) & Ho, Tai-Hau & 0.007538 & 0.007732 & 0.007930\\
10-12) & Hayakawa, Sessue & 0.007399 & 0.007593 & 0.007792\\
10-12) & Haas, Hugo (I) & 0.007158 & 0.007352 & 0.007552\\
\hline
\end{tabular}
\end{table}

\begin{table}[H]
\caption{The top-$k$ betweenness centralities of a snapshot of the IMDB collaboration network taken at the end of 1954 (120430 nodes), computed by \cad with $\delta=0.1$ and $\lambda = 0.0002$.}
\centering
\begin{tabular}{|l|l|r|r|r|}
\hline
Ranking & Actor & Lower bound & Estimated betweenness & Upper bound \\ 
\hline
1) & Meyer, Torben & 0.013418 & 0.013868 & 0.014334\\
2) & Kamiyama, S\=ojin & 0.010331 & 0.010726 & 0.011089\\
3-4) & Ertugrul, Muhsin & 0.009956 & 0.010141 & 0.010331\\
3-4) & Jung, Shia & 0.009643 & 0.009826 & 0.010013\\
5-6) & Singh, Ram (I) & 0.008657 & 0.008841 & 0.009030\\
5-6) & Paananen, Tuulikki & 0.008383 & 0.008567 & 0.008755\\
7-9) & Par\`is, Manuel & 0.007886 & 0.008070 & 0.008257\\
7-10) & Goetzke, Bernhard & 0.007802 & 0.007987 & 0.008176\\
7-10) & Yamaguchi, Shirley & 0.007531 & 0.007716 & 0.007905\\
8-10) & Hayakawa, Sessue & 0.007473 & 0.007657 & 0.007845\\
\hline
\end{tabular}
\end{table}

\begin{table}[H]
\caption{The top-$k$ betweenness centralities of a snapshot of the IMDB collaboration network taken at the end of 1959 (146253 nodes), computed by \cad with $\delta=0.1$ and $\lambda = 0.0002$.}
\centering
\begin{tabular}{|l|l|r|r|r|}
\hline
Ranking & Actor & Lower bound & Estimated betweenness & Upper bound \\ 
\hline
1-2) & Singh, Ram (I) & 0.010683 & 0.010877 & 0.011075\\
1-2) & Frees, Paul & 0.010372 & 0.010566 & 0.010763\\
3) & Meyer, Torben & 0.009478 & 0.009821 & 0.010235\\
4-5) & Jung, Shia & 0.008623 & 0.008816 & 0.009013\\
4-5) & Ghosh, Sachin & 0.008459 & 0.008651 & 0.008847\\
6-7) & Myzet, Rudolf & 0.007085 & 0.007278 & 0.007476\\
6-7) & Yamaguchi, Shirley & 0.006908 & 0.007101 & 0.007299\\
8) & de C\`ordova, Arturo & 0.006391 & 0.006582 & 0.006778\\
9-11) & Kamiyama, S\=ojin & 0.005861 & 0.006054 & 0.006254\\
9-12) & Paananen, Tuulikki & 0.005810 & 0.006003 & 0.006202\\
9-12) & Flowers, Bess & 0.005620 & 0.005813 & 0.006012\\
10-12) & Par\`is, Manuel & 0.005442 & 0.005635 & 0.005835\\
\hline
\end{tabular}
\end{table}

\begin{table}[H]
\caption{The top-$k$ betweenness centralities of a snapshot of the IMDB collaboration network taken at the end of 1964 (174826 nodes), computed by \cad with $\delta=0.1$ and $\lambda = 0.0002$.}
\centering
\begin{tabular}{|l|l|r|r|r|}
\hline
Ranking & Actor & Lower bound & Estimated betweenness & Upper bound \\ 
\hline
1) & Frees, Paul & 0.013140 & 0.013596 & 0.014067\\
2) & Meyer, Torben & 0.007279 & 0.007617 & 0.007856\\
3-4) & Harris, Sam (II) & 0.006813 & 0.006967 & 0.007124\\
3-5) & Myzet, Rudolf & 0.006696 & 0.006849 & 0.007005\\
4-5) & Flowers, Bess & 0.006422 & 0.006572 & 0.006726\\
6) & Kong, King (I) & 0.005909 & 0.006104 & 0.006422\\
7) & Yuen, Siu Tin & 0.005114 & 0.005264 & 0.005420\\
8) & Miller, Marvin (I) & 0.004708 & 0.004859 & 0.005015\\
9-12) & de C\`ordova, Arturo & 0.004147 & 0.004299 & 0.004457\\
9-18) & Haas, Hugo (I) & 0.003888 & 0.004039 & 0.004197\\
9-18) & Singh, Ram (I) & 0.003854 & 0.004004 & 0.004160\\
9-18) & Kamiyama, S\=ojin & 0.003848 & 0.003999 & 0.004155\\
10-18) & Sauli, Anneli & 0.003827 & 0.003978 & 0.004135\\
10-18) & King, Walter Woolf & 0.003774 & 0.003923 & 0.004078\\
10-18) & Vanel, Charles & 0.003716 & 0.003867 & 0.004024\\
10-18) & Kowall, Mitchell & 0.003684 & 0.003834 & 0.003990\\
10-18) & Holmes, Stuart & 0.003603 & 0.003752 & 0.003907\\
10-18) & Sten, Anna & 0.003582 & 0.003733 & 0.003890\\
\hline
\end{tabular}
\end{table}

\begin{table}[H]
\caption{The top-$k$ betweenness centralities of a snapshot of the IMDB collaboration network taken at the end of 1969 (210527 nodes), computed by \cad with $\delta=0.1$ and $\lambda = 0.0002$.}
\centering
\begin{tabular}{|l|l|r|r|r|}
\hline
Ranking & Actor & Lower bound & Estimated betweenness & Upper bound \\ 
\hline
1) & Frees, Paul & 0.010913 & 0.011446 & 0.012005\\
2-3) & Yuen, Siu Tin & 0.006157 & 0.006349 & 0.006547\\
2-3) & Tamiroff, Akim & 0.006097 & 0.006291 & 0.006490\\
4-6) & Meyer, Torben & 0.005675 & 0.005869 & 0.006069\\
4-7) & Harris, Sam (II) & 0.005639 & 0.005830 & 0.006027\\
4-8) & Rubener, Sujata & 0.005427 & 0.005618 & 0.005815\\
5-8) & Myzet, Rudolf & 0.005253 & 0.005444 & 0.005641\\
6-8) & Flowers, Bess & 0.005136 & 0.005328 & 0.005526\\
9-10) & Kong, King (I) & 0.004354 & 0.004544 & 0.004741\\
9-10) & Sullivan, Elliott & 0.004208 & 0.004398 & 0.004596\\
\hline
\end{tabular}
\end{table}

\begin{table}[H]
\caption{The top-$k$ betweenness centralities of a snapshot of the IMDB collaboration network taken at the end of 1974 (257896 nodes), computed by \cad with $\delta=0.1$ and $\lambda = 0.0002$.}
\centering
\begin{tabular}{|l|l|r|r|r|}
\hline
Ranking & Actor & Lower bound & Estimated betweenness & Upper bound \\ 
\hline
1) & Frees, Paul & 0.008507 & 0.008958 & 0.009295\\
2) & Chen, Sing & 0.007734 & 0.008056 & 0.008507\\
3) & Welles, Orson & 0.006115 & 0.006497 & 0.006903\\
4-5) & Loren, Sophia & 0.005056 & 0.005221 & 0.005392\\
4-7) & Rubener, Sujata & 0.004767 & 0.004933 & 0.005106\\
5-8) & Harris, Sam (II) & 0.004628 & 0.004795 & 0.004967\\
5-8) & Tamiroff, Akim & 0.004625 & 0.004790 & 0.004962\\
6-10) & Meyer, Torben & 0.004382 & 0.004548 & 0.004720\\
8-12) & Flowers, Bess & 0.004259 & 0.004425 & 0.004598\\
8-12) & Yuen, Siu Tin & 0.004229 & 0.004397 & 0.004571\\
9-12) & Carradine, John & 0.004026 & 0.004192 & 0.004364\\
9-12) & Myzet, Rudolf & 0.003984 & 0.004151 & 0.004325\\
\hline
\end{tabular}
\end{table}

\begin{table}[H]
\caption{The top-$k$ betweenness centralities of a snapshot of the IMDB collaboration network taken at the end of 1979 (310278 nodes), computed by \cad with $\delta=0.1$ and $\lambda = 0.0002$.}
\centering
\begin{tabular}{|l|l|r|r|r|}
\hline
Ranking & Actor & Lower bound & Estimated betweenness & Upper bound \\ 
\hline
1) & Chen, Sing & 0.007737 & 0.008220 & 0.008647\\
2) & Frees, Paul & 0.006852 & 0.007255 & 0.007737\\
3-5) & Welles, Orson & 0.004894 & 0.005075 & 0.005263\\
3-6) & Carradine, John & 0.004623 & 0.004803 & 0.004989\\
3-6) & Loren, Sophia & 0.004614 & 0.004796 & 0.004985\\
4-6) & Rubener, Sujata & 0.004284 & 0.004464 & 0.004651\\
7-17) & Tamiroff, Akim & 0.003516 & 0.003696 & 0.003885\\
7-17) & Meyer, Torben & 0.003479 & 0.003657 & 0.003844\\
7-17) & Quinn, Anthony (I) & 0.003447 & 0.003626 & 0.003815\\
7-17) & Flowers, Bess & 0.003446 & 0.003625 & 0.003815\\
7-17) & Mitchell, Gordon (I) & 0.003417 & 0.003596 & 0.003785\\
7-17) & Sullivan, Elliott & 0.003371 & 0.003551 & 0.003740\\
7-17) & Rietty, Robert & 0.003368 & 0.003547 & 0.003735\\
7-17) & Tanba, Tetsur\=o & 0.003360 & 0.003537 & 0.003724\\
7-17) & Harris, Sam (II) & 0.003331 & 0.003510 & 0.003699\\
7-17) & Lewgoy, Jos\`e & 0.003223 & 0.003402 & 0.003590\\
7-17) & Dalio, Marcel & 0.003185 & 0.003364 & 0.003553\\
\hline
\end{tabular}
\end{table}

\begin{table}[H]
\caption{The top-$k$ betweenness centralities of a snapshot of the IMDB collaboration network taken at the end of 1984 (375322 nodes), computed by \cad with $\delta=0.1$ and $\lambda = 0.0002$.}
\centering
\begin{tabular}{|l|l|r|r|r|}
\hline
Ranking & Actor & Lower bound & Estimated betweenness & Upper bound \\ 
\hline
1) & Chen, Sing & 0.007245 & 0.007716 & 0.008218\\
2-4) & Welles, Orson & 0.005202 & 0.005391 & 0.005587\\
2-4) & Frees, Paul & 0.005174 & 0.005363 & 0.005559\\
2-5) & Hitler, Adolf & 0.004906 & 0.005094 & 0.005290\\
4-6) & Carradine, John & 0.004744 & 0.004932 & 0.005127\\
5-7) & Mitchell, Gordon (I) & 0.004418 & 0.004606 & 0.004802\\
6-8) & J\"urgens, Curd & 0.004169 & 0.004356 & 0.004551\\
7-8) & Kinski, Klaus & 0.003938 & 0.004123 & 0.004318\\
9-12) & Rubener, Sujata & 0.003396 & 0.003585 & 0.003785\\
9-12) & Lee, Christopher (I) & 0.003391 & 0.003576 & 0.003771\\
9-12) & Loren, Sophia & 0.003357 & 0.003542 & 0.003738\\
9-12) & Harrison, Richard (II) & 0.003230 & 0.003417 & 0.003614\\
\hline
\end{tabular}
\end{table}

\begin{table}[H]
\caption{The top-$k$ betweenness centralities of a snapshot of the IMDB collaboration network taken at the end of 1989 (463078 nodes), computed by \cad with $\delta=0.1$ and $\lambda = 0.0002$.}
\centering
\begin{tabular}{|l|l|r|r|r|}
\hline
Ranking & Actor & Lower bound & Estimated betweenness & Upper bound \\ 
\hline
1-2) & Hitler, Adolf & 0.005282 & 0.005467 & 0.005658\\
1-3) & Chen, Sing & 0.005008 & 0.005192 & 0.005382\\
2-4) & Carradine, John & 0.004648 & 0.004834 & 0.005027\\
3-4) & Harrison, Richard (II) & 0.004515 & 0.004697 & 0.004887\\
5-6) & Welles, Orson & 0.004088 & 0.004271 & 0.004462\\
5-9) & Mitchell, Gordon (I) & 0.003766 & 0.003948 & 0.004139\\
6-9) & Kinski, Klaus & 0.003691 & 0.003874 & 0.004065\\
6-11) & Lee, Christopher (I) & 0.003610 & 0.003793 & 0.003984\\
6-11) & Frees, Paul & 0.003582 & 0.003766 & 0.003960\\
8-13) & J\"urgens, Curd & 0.003306 & 0.003486 & 0.003676\\
8-13) & Pleasence, Donald & 0.003299 & 0.003479 & 0.003670\\
10-13) & Mitchell, Cameron (I) & 0.003105 & 0.003285 & 0.003476\\
10-13) & von Sydow, Max (I) & 0.002982 & 0.003161 & 0.003350\\
\hline
\end{tabular}
\end{table}

\begin{table}[H]
\caption{The top-$k$ betweenness centralities of a snapshot of the IMDB collaboration network taken at the end of 1994 (557373 nodes), computed by \cad with $\delta=0.1$ and $\lambda = 0.0002$.}
\centering
\begin{tabular}{|l|l|r|r|r|}
\hline
Ranking & Actor & Lower bound & Estimated betweenness & Upper bound \\ 
\hline
1) & Hitler, Adolf & 0.005227 & 0.005676 & 0.006164\\
2-6) & Harrison, Richard (II) & 0.003978 & 0.004165 & 0.004362\\
2-6) & von Sydow, Max (I) & 0.003884 & 0.004069 & 0.004264\\
2-7) & Lee, Christopher (I) & 0.003718 & 0.003907 & 0.004106\\
2-7) & Carradine, John & 0.003696 & 0.003883 & 0.004079\\
2-7) & Chen, Sing & 0.003683 & 0.003871 & 0.004068\\
4-10) & Jeremy, Ron & 0.003336 & 0.003524 & 0.003722\\
7-11) & Pleasence, Donald & 0.003253 & 0.003439 & 0.003637\\
7-11) & Rey, Fernando (I) & 0.003234 & 0.003420 & 0.003617\\
7-15) & Smith, William (I) & 0.003012 & 0.003199 & 0.003397\\
8-15) & Welles, Orson & 0.002885 & 0.003072 & 0.003271\\
10-15) & Mitchell, Gordon (I) & 0.002851 & 0.003036 & 0.003232\\
10-15) & Kinski, Klaus & 0.002705 & 0.002890 & 0.003087\\
10-15) & Mitchell, Cameron (I) & 0.002671 & 0.002858 & 0.003058\\
10-15) & Quinn, Anthony (I) & 0.002640 & 0.002826 & 0.003026\\
\hline
\end{tabular}
\end{table}

\begin{table}[H]
\caption{The top-$k$ betweenness centralities of a snapshot of the IMDB collaboration network taken at the end of 1999 (681358 nodes), computed by \cad with $\delta=0.1$ and $\lambda = 0.0002$.}
\centering
\begin{tabular}{|l|l|r|r|r|}
\hline
Ranking & Actor & Lower bound & Estimated betweenness & Upper bound \\ 
\hline
1) & Jeremy, Ron & 0.007380 & 0.007913 & 0.008484\\
2) & Hitler, Adolf & 0.004601 & 0.005021 & 0.005480\\
3-4) & Lee, Christopher (I) & 0.003679 & 0.003849 & 0.004028\\
3-4) & von Sydow, Max (I) & 0.003604 & 0.003775 & 0.003953\\
5-6) & Harrison, Richard (II) & 0.003041 & 0.003211 & 0.003390\\
5-7) & Carradine, John & 0.002943 & 0.003114 & 0.003296\\
6-11) & Chen, Sing & 0.002662 & 0.002834 & 0.003018\\
7-14) & Rey, Fernando (I) & 0.002569 & 0.002740 & 0.002922\\
7-14) & Smith, William (I) & 0.002559 & 0.002729 & 0.002910\\
7-14) & Pleasence, Donald & 0.002556 & 0.002725 & 0.002906\\
7-14) & Sutherland, Donald (I) & 0.002449 & 0.002617 & 0.002796\\
8-14) & Quinn, Anthony (I) & 0.002307 & 0.002476 & 0.002658\\
8-14) & Mastroianni, Marcello & 0.002271 & 0.002440 & 0.002621\\
8-14) & Saxon, John & 0.002251 & 0.002420 & 0.002602\\
\hline
\end{tabular}
\end{table}

\begin{table}[H]
\caption{The top-$k$ betweenness centralities of a snapshot of the IMDB collaboration network taken at the end of 2004 (880032 nodes), computed by \cad with $\delta=0.1$ and $\lambda = 0.0002$.}
\centering
\begin{tabular}{|l|l|r|r|r|}
\hline
Ranking & Actor & Lower bound & Estimated betweenness & Upper bound \\ 
\hline
1) & Jeremy, Ron & 0.010653 & 0.011370 & 0.012136\\
2) & Hitler, Adolf & 0.005333 & 0.005840 & 0.006396\\
3-4) & von Sydow, Max (I) & 0.003424 & 0.003608 & 0.003802\\
3-4) & Lee, Christopher (I) & 0.003403 & 0.003587 & 0.003781\\
5-6) & Kier, Udo & 0.002898 & 0.003081 & 0.003275\\
5-8) & Keitel, Harvey (I) & 0.002646 & 0.002828 & 0.003023\\
6-12) & Hopper, Dennis & 0.002424 & 0.002607 & 0.002804\\
6-16) & Smith, William (I) & 0.002322 & 0.002504 & 0.002700\\
7-17) & Sutherland, Donald (I) & 0.002241 & 0.002422 & 0.002617\\
7-23) & Carradine, David & 0.002149 & 0.002329 & 0.002526\\
7-23) & Carradine, John & 0.002147 & 0.002328 & 0.002524\\
7-23) & Harrison, Richard (II) & 0.002054 & 0.002234 & 0.002430\\
8-23) & Sharif, Omar & 0.002043 & 0.002222 & 0.002418\\
8-23) & Steiger, Rod & 0.001988 & 0.002165 & 0.002358\\
8-23) & Quinn, Anthony (I) & 0.001974 & 0.002151 & 0.002344\\
8-23) & Depardieu, G\`erard & 0.001966 & 0.002148 & 0.002346\\
9-23) & Sheen, Martin & 0.001913 & 0.002093 & 0.002291\\
10-23) & Rey, Fernando (I) & 0.001866 & 0.002044 & 0.002238\\
10-23) & Kane, Sharon & 0.001857 & 0.002038 & 0.002237\\
10-23) & Pleasence, Donald & 0.001859 & 0.002037 & 0.002232\\
10-23) & Skarsg\.{a}rd, Stellan & 0.001848 & 0.002026 & 0.002221\\
10-23) & Mueller-Stahl, Armin & 0.001789 & 0.001969 & 0.002166\\
10-23) & Hong, James (I) & 0.001780 & 0.001957 & 0.002152\\
\hline
\end{tabular}
\end{table}

\begin{table}[H]
\caption{The top-$k$ betweenness centralities of a snapshot of the IMDB collaboration network taken at the end of 2009 (1237879 nodes), computed by \cad with $\delta=0.1$ and $\lambda = 0.0002$.}
\centering
\begin{tabular}{|l|l|r|r|r|}
\hline
Ranking & Actor & Lower bound & Estimated betweenness & Upper bound \\ 
\hline
1) & Jeremy, Ron & 0.010531 & 0.011237 & 0.011991\\
2) & Hitler, Adolf & 0.005500 & 0.006011 & 0.006568\\
3-4) & Kaufman, Lloyd & 0.003620 & 0.003804 & 0.003997\\
3-4) & Kier, Udo & 0.003472 & 0.003654 & 0.003845\\
5-6) & Lee, Christopher (I) & 0.003056 & 0.003240 & 0.003435\\
5-8) & Carradine, David & 0.002866 & 0.003050 & 0.003245\\
6-8) & Keitel, Harvey (I) & 0.002659 & 0.002840 & 0.003034\\
6-9) & von Sydow, Max (I) & 0.002532 & 0.002713 & 0.002907\\
8-13) & Hopper, Dennis & 0.002237 & 0.002419 & 0.002616\\
9-15) & Skarsg\.{a}rd, Stellan & 0.002153 & 0.002333 & 0.002529\\
9-15) & Depardieu, G\`erard & 0.002001 & 0.002181 & 0.002377\\
9-15) & Hauer, Rutger & 0.001894 & 0.002074 & 0.002271\\
9-15) & Sutherland, Donald (I) & 0.001875 & 0.002054 & 0.002250\\
10-15) & Smith, William (I) & 0.001811 & 0.001990 & 0.002186\\
10-15) & Dafoe, Willem & 0.001805 & 0.001986 & 0.002186\\
\hline
\end{tabular}
\end{table}

\begin{table}[H]
\caption{The top-$k$ betweenness centralities of a snapshot of the IMDB collaboration network taken in 2014 (1797446
 nodes), computed by \cad with $\delta=0.1$ and $\lambda = 0.0002$.}
\centering
\begin{tabular}{|l|l|r|r|r|}
\hline
Ranking & Actor & Lower bound & Estimated betweenness & Upper bound \\ 
\hline
1) & Jeremy, Ron & 0.009360 & 0.010058 & 0.010808\\
2) & Kaufman, Lloyd & 0.005936 & 0.006492 & 0.007100\\
3) & Hitler, Adolf & 0.004368 & 0.004844 & 0.005373\\
4-6) & Kier, Udo & 0.003250 & 0.003435 & 0.003631\\
4-6) & Roberts, Eric (I) & 0.003178 & 0.003362 & 0.003557\\
4-6) & Madsen, Michael (I) & 0.003120 & 0.003305 & 0.003501\\
7-9) & Trejo, Danny & 0.002652 & 0.002835 & 0.003030\\
7-9) & Lee, Christopher (I) & 0.002551 & 0.002734 & 0.002931\\
7-12) & Estevez, Joe & 0.002350 & 0.002534 & 0.002732\\
9-17) & Carradine, David & 0.002116 & 0.002296 & 0.002492\\
9-17) & von Sydow, Max (I) & 0.002023 & 0.002206 & 0.002405\\
9-17) & Keitel, Harvey (I) & 0.001974 & 0.002154 & 0.002352\\
10-17) & Skarsg\.{a}rd, Stellan & 0.001945 & 0.002125 & 0.002323\\
10-17) & Dafoe, Willem & 0.001899 & 0.002080 & 0.002279\\
10-17) & Hauer, Rutger & 0.001891 & 0.002071 & 0.002269\\
10-17) & Depardieu, G\`erard & 0.001763 & 0.001943 & 0.002142\\
10-17) & Rochon, Debbie & 0.001745 & 0.001926 & 0.002126\\
\hline
\end{tabular}
\end{table}

\end{document}